\documentclass[opre,nonblindrev]{informs3}
\DoubleSpacedXI 
\newcount\Comments
\newcommand{\kibitz}[2]{\ifnum\Comments=1{\color{#1}{#2}}\fi}

\newcount\Drop  
\newcommand{\todrop}[2]{\ifnum\Drop=1{\color{#1}{#2}}\fi}

\newif\ifhighlight 
\highlighttrue 

\newif\iftodo

\usepackage{amsmath, bm}

\newcommand\vu{{\bm u}}
\def\vv{{\bm v}}
\newcommand\vw{{\bm w}}

\newcommand{\R}{\mathbb{R}}

\renewcommand{\P}{\mathbb{P}}

\newcommand{\E}{\mathbb{E}}

\newcommand{\indicator}{\mathds{1}}

\newcommand{\defeq}{\triangleq}

\newcommand{\gre}{\mathsf{GRE}}
\newcommand{\opt}{\mathsf{OPT}}
\newcommand{\off}{\mathsf{OFF}}

\usepackage[english]{babel}
\usepackage[autostyle, english = american]{csquotes}
\MakeOuterQuote{"}

\usepackage{bm,bbm,xspace,multirow,multicol,dsfont}
\usepackage[colorlinks,
	citecolor = dark6,
	urlcolor = black,
	linkcolor = dark6,
    hypertexnames=false]{hyperref}
    
\usepackage[font=scriptsize]{subcaption}
\usepackage[short]{optidef}
\usepackage[nameinlink]{cleveref}
\crefname{subsection}{subsection}{subsections}
\usepackage[normalem]{ulem}
\usepackage{algorithm}
\usepackage{algpseudocode}
\algrenewcommand\algorithmicrequire{\textbf{Input:}}
\algrenewcommand\algorithmicensure{\textbf{Output:}}
\usepackage{tikz,pgfplots}
\pgfplotsset{compat=1.18}  

\def\theARTICLETOP{}
\def\theLRHFirstLine{}
\def\theLRHSecondLine{}
\def\theRRHFirstLine{}
\def\theRRHSecondLine{}

\usepackage{xcolor}
\definecolor{dark1}{RGB}{157, 58, 103}
\definecolor{dark2}{RGB}{161, 67, 0}
\definecolor{dark3}{RGB}{115, 102, 0}
\definecolor{dark4}{RGB}{2, 120, 50}
\definecolor{dark5}{RGB}{0, 116, 122}
\definecolor{dark6}{RGB}{18, 100, 176}
\definecolor{dark7}{RGB}{116, 75, 163}
\definecolor{mid1}{RGB}{225, 119, 163}
\definecolor{mid2}{RGB}{228, 128, 77}
\definecolor{mid3}{RGB}{182, 158, 21}
\definecolor{mid4}{RGB}{87, 182, 109}
\definecolor{mid5}{RGB}{0, 181, 190}
\definecolor{mid6}{RGB}{88, 162, 242}
\definecolor{mid7}{RGB}{177, 135, 229}
\definecolor{light1}{RGB}{255, 187, 231}
\definecolor{light2}{RGB}{255, 198, 151}
\definecolor{light3}{RGB}{247, 223, 104}
\definecolor{light4}{RGB}{152, 248, 171}
\definecolor{light5}{RGB}{72, 249, 255}
\definecolor{light6}{RGB}{164, 219, 255}
\definecolor{light7}{RGB}{244, 192, 255}
\definecolor{lightyellow}{RGB}{255, 255, 204}

\usepackage{color}              
\usepackage{color-edits}
\addauthor{Will}{blue}
\addauthor{Hma}{dark4}

\usepackage{natbib}
 \bibpunct[, ]{(}{)}{,}{a}{}{,}%
 \def\bibfont{\small}%
\TheoremsNumberedThrough     
\ECRepeatTheorems

\EquationsNumberedThrough    

\usepackage{enumerate}

\begin{document}


\RUNAUTHOR{Ma, Ma, and Ro mero}

\RUNTITLE{Greedy Sequential Offering}

\TITLE{Sequential Offering in On-Demand Platforms: On the Optimality of Greedy Ranking}

\ARTICLEAUTHORS{
\AUTHOR{Hongyao Ma}
\AFF{Graduate School of Business, Columbia University, New York, NY 10027, \EMAIL{hongyao.ma@columbia.edu}}
\AUTHOR{Will Ma}
\AFF{Graduate School of Business, Columbia University, New York, NY 10027, \EMAIL{wm2428@gsb.columbia.edu}}
\AUTHOR{Matias Romero}
\AFF{Graduate School of Business, Columbia University, New York, NY 10027, \EMAIL{mer2262@columbia.edu}}
}

\ABSTRACT{
On-demand platforms face the fundamental challenge of fulfilling time-sensitive jobs with independent workers who may decline offers.
To minimize delays and unfulfilled jobs, platforms frequently raise the offered wage sequentially following each rejection.
However, the interaction between these dynamic price adjustments and the specific sequence in which workers are approached has been overlooked.
In particular, if the best-suited workers (e.g., closest to the job) are also ranked earliest in the sequence, then those workers would see the lowest offered wages and may decline, leading to poor system outcomes where less-suited workers end up seeing the raised wages and accepting the job.

We study the sequential offering problem to maximize expected welfare or platform profit by jointly optimizing the ranking of workers and the pricing trajectory.
Surprisingly, our main result establishes that if the reservation wage distribution exhibits a non-increasing and convex density function (e.g., Uniform, Exponential), welfare is maximized by greedy ranking and wages optimized via backward induction.
For arbitrary distributions, we prove that greedy ranking achieves a tight $n/(2n-1)$ fraction of the prophet benchmark.
Numerical results for settings beyond the distributional assumptions find welfare losses well below those allowed by the universal guarantee, even in families where greedy is provably suboptimal.
This suggests that rather than sending initial ``low ball'' offers to worse
matches, platforms should stick with greedy ranking and optimize the wage offerings by appropriately
taking the continuation value of the downstream offers into consideration.
}



\maketitle

\section{Introduction} \label{sec:intro}

On-demand platforms face the challenge of fulfilling time-sensitive jobs (i.e., customer orders or requests) without the authority to enforce acceptance from workers (e.g. drivers, couriers, etc).
Instead, a job is offered to workers sequentially, and is often rejected multiple times before a worker accepts the job.
These rejections are largely driven by workers' sensitivity to offered wages \citep{allon2023impact}.
With job characteristics and wage now shown upfront across major platforms, the resulting cherry-picking has led to widespread reports of single-digit acceptance rates \citep{ford2021bloomberg,reddit_low_acceptance}.
%

To reduce delays and unfulfilled jobs, both of which entail significant costs, platforms often increase the offered wage after each rejection.
This strategy has been widely adopted across diverse domains, spanning food delivery \citep{doordash2025earnings}, ride-sharing \citep{uber2024earnings}, freight marketplaces \citep{amazon2023flexoffers}, and healthcare staffing \citep{carerev2023smartrates}.
While such dynamic wage adjustments effectively improve reliability, the interaction with the sequence (i.e., ranking) in which workers are chosen has been largely overlooked. 
Specifically, greedily ranking the best worker first can be inefficient when the offered wages first “low ball” the worker who are good matches, and later incentivize the poor matches to accept the job with higher wages. 

To illustrate, consider offering a job sequentially to two workers, where one is a better match than the other.
Suppose the platform makes an initial ``low ball'' offer of \$5 that is always rejected.
If rejected, the wage increases for the second worker to \$10, which is accepted by most workers.
Under a greedy ranking, the best-matching worker rejects the first offer and the worse-matching worker fulfills the job.


This example raises the question of whether platforms should adjust their ranking decisions in response to anticipated wage sequences.
%
However, attempting to solve this joint optimization problem iteratively creates a complex feedback loop that is not guaranteed to resolve the underlying inefficiency.

In this paper, we propose the following simple method for jointly optimizing the ranking and wages: fix the ranking to be greedy, and then optimize wages conditional on that ranking.
This replaces a search over worker permutations with a single sorting step, ranking workers in decreasing order of their match quality with the job.
Given this ranking, the optimal wages can then be tractably computed using dynamic programming.

%

Surprisingly, we show that this simple greedy-rank-then-optimize-wages approach is in fact optimal for a broad class of distributions governing workers' willingness to work. 
We next formalize the model and present our main results.

\subsection{Model Description}
We consider a platform seeking to assign a job to one of several heterogeneous workers.
Worker $i$ generates value $v_i$ to the platform if assigned the job, where differences in $v_i$ may reflect, for example, workers’ locations and hence their quality of match with the job. 
The platform approaches workers sequentially, choosing both the sequence in which to approach them and the wage offered to each worker; worker $i$ accepts if the offered wage is at least their reservation wage, drawn independently from a common distribution $F$. 
The platform’s objective is to jointly choose the ranking and wages to maximize its expected payoff.
Our primary results concern total welfare, the match value minus the worker's realized private cost.
Even for this objective, unlimited wages are undesirable: selective offers preserve the opportunity to obtain greater welfare from another worker.
Optimal net wages account for this opportunity cost and need not increase after every rejection, while the corresponding total wages weakly increase as the continuation value falls.
We also demonstrate in \Cref{sec:quantile_framework} that the optimality of greedy ranking can be extended to profit maximization for a restricted class of distributions that includes Exponential and Uniform.

\subsection{Main Contributions}

\paragraph{Exact optimality of greedy ranking.} 
We establish in \Cref{thm:global_optimality} that if the distribution of reservation wages admits a non-increasing and convex density function, the exact welfare-optimal solution is achieved by greedy ranking and dynamically optimizing wages via backward induction.
A non-increasing density means that an additional dollar raises acceptance less at higher wages, while convexity means that this marginal response declines at non-increasing rate.
The first property keeps earlier, selectively priced offers sufficiently likely to succeed, and the second controls the welfare effect of re-optimizing wages after a swap.
The class includes Uniform, Exponential and Pareto distributions.
We also prove greedy optimality for Half-normal costs, despite their non-convex density (\Cref{sec:half_normal}), while Logistic costs give a distinct result: every ranking yields the same welfare under its optimal wages (\Cref{sec:logistic_indifference}), and hence greedy is still optimal.
%

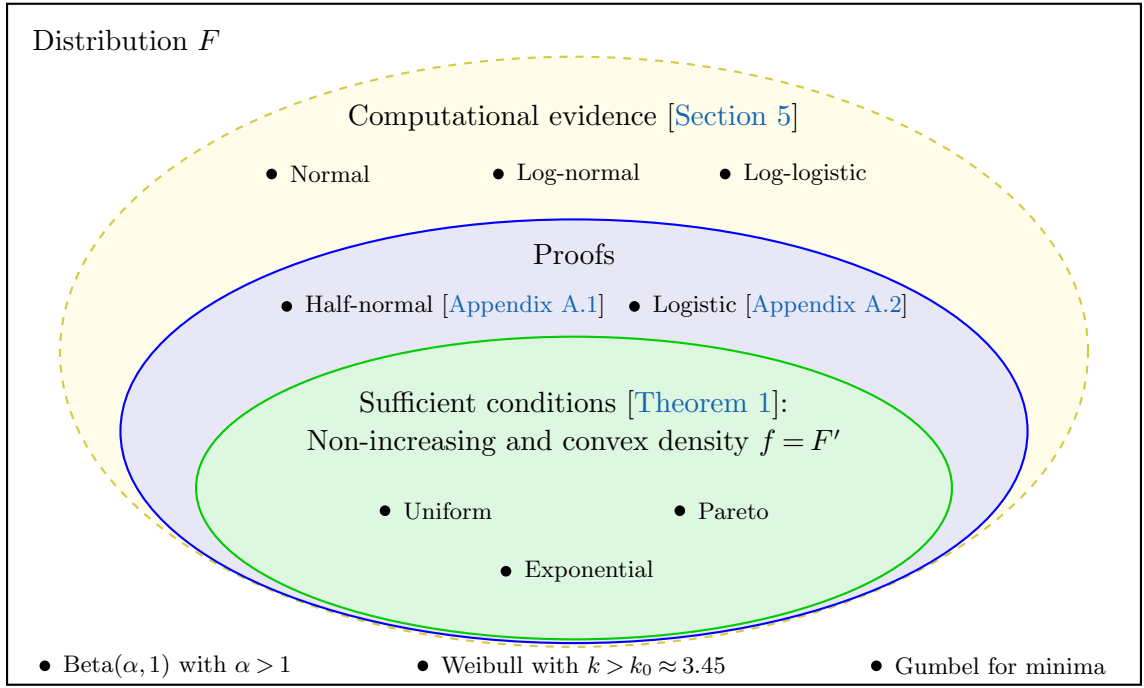
\begin{figure}[htbp]
    \centering
    \begin{tikzpicture}
        \draw[thick] (-7.5, -4.5) rectangle (7.5, 4.5);
        \node[anchor=north west] at (-7.3, 4.3) {Distribution $F$};

        \fill[yellow!20, fill opacity=0.5] (0, -0.1) ellipse (6.8cm and 3.9cm);
        \draw[yellow!80!black, thick, dashed] (0, -0.1) ellipse (6.8cm and 3.9cm);
        \node at (0, 3.0) {Computational evidence [\Cref{sec:simulations}]};

        \fill[blue!15, fill opacity=0.6] (0, -1.15) ellipse (6.0cm and 2.8cm);
        \draw[blue, thick] (0, -1.15) ellipse (6.0cm and 2.8cm);
        \node at (0, 1.2) {Proofs};

        \fill[green!15, fill opacity=0.7] (0, -1.9) ellipse (5.0cm and 2.0cm);
        \draw[green!80!black, thick] (0, -1.9) ellipse (5.0cm and 2.0cm);
        \node at (0, -0.8) {Sufficient conditions [\Cref{thm:global_optimality}]:
        };
        \node at (0, -1.3) {Non-increasing and convex density $f=F'$
        };

        
        \fill (-7.0, -4.25) circle (2pt);
        \node[right, align=left] at (-7.0, -4.25) {\footnotesize Beta$(\alpha,1)$ with $\alpha>1$
        };
        \fill (-2.0, -4.25) circle (2pt);
        \node[right, align=left] at (-2.0, -4.25) {\footnotesize Weibull with $k>k_0\approx3.45$
        };
        \fill (4.0, -4.25) circle (2pt);
        \node[right, align=left] at (4.0, -4.25) {\footnotesize Gumbel for minima
        };

        \fill (-4, 2.25) circle (2pt);
        \node[right] at (-4, 2.25) {\footnotesize Normal};
        \fill (-1, 2.25) circle (2pt);
        \node[right] at (-1, 2.25) {\footnotesize Log-normal};
        \fill (2., 2.25) circle (2pt);
        \node[right] at (2., 2.25) {\footnotesize Log-logistic};

        \fill (0.8, 0.5) circle (2pt);
        \node[right] at (0.8, 0.5) {\footnotesize Logistic [\Cref{sec:logistic_indifference}]};
        \fill (-3.8, 0.5) circle (2pt);
        \node[right] at (-3.8, 0.5) {\footnotesize Half-normal [\Cref{sec:half_normal}]};
        

        \fill (-2.5, -2.2) circle (2pt);
        \node[right] at (-2.5, -2.2) {\footnotesize Uniform};
        \fill (-0.9, -3.0) circle (2pt);
        \node[right] at (-0.9, -3.0) {\footnotesize Exponential};
        \fill (1.4, -2.2) circle (2pt);
        \node[right] at (1.4, -2.2) {\footnotesize Pareto};

    \end{tikzpicture}
    \caption{Selected reservation wage distributions classified by the evidence for greedy optimality. The distribution below the ellipses admit counterexamples}
    \label{fig:distributions_hierarchy_colored}
\end{figure}

\paragraph{Worst-case guarantees.}
To formalize the performance limits of the greedy heuristic when it is suboptimal, we establish tight worst-case guarantees.
We know from the prophet inequality literature and its connection to sequential pricing \citep{correa2019pricing} that $1/2$ serves as a universal lower bound for greedy ranking relative to the offline optimum.
We show that greedy ranking cannot do better even if it is only comparing against the best ranking, noting that we need a new example because the basic counterexample of 1/2 \citep{krengel1977semiamarts} does not fit into our IID setting.
Moreover, we show for every number of workers $n$ that the tight performance guarantee of greedy is $n/(2n-1)$, which approaches 1/2 as $n\to\infty$.

\paragraph{Numerical performance beyond sufficient conditions.} 
Extensive numerical results find no counterexamples to greedy optimality for Normal, Log-normal, Log-logistic, and Gumbel for maxima (\Cref{fig:distributions_hierarchy_colored}).
More revealing are the cases where greedy is suboptimal.
A mean-below-median condition produces counterexamples for Beta, Weibull and Gumbel for minima.
Yet deliberate searches for adverse instances find much smaller losses than the universal guarantee permits.
At eight workers, the largest losses found are about 2.4\% for Weibull with $k=10$, 5.0\% for Gumbel (for minima), and 11.4\% for Beta, compared with the worst-case loss of 46.7\% (\Cref{sec:simulations}).

\subsection{Connection to an open theoretical problem in prophet inequalities}
\label{sec:prophet_inequalities}

Our joint optimization problem can be framed as a \textit{free-order} prophet inequality problem \citep[see][for problem definition]{hill1983prophet}.
In the general formulation, determining the optimal sequencing rule is computationally intractable and known to be NP-hard \citep{agrawal2020optimal}.
However, our environment exhibits a specific structure: rather than the general problem, we are studying an IID special case with location shifts, where each worker's acceptance threshold is determined by a heterogeneous parameter.

Because of this structure, we ask a very different question, which is \textit{characterizing the set of distributions} for which the greedy-rank-then-optimize-wages heuristic is \textit{optimal} under \textit{any} location shifts. 
A similar question has in fact previously intrigued \citet{hill1985selection}; however, to the best of our knowledge, our paper represents the only progress on this type of question since then.

\subsection{Further Related Work}

\paragraph{Matching and pricing in on-demand platforms.} 
Our work contributes to the extensive literature on pricing and matching in on-demand platforms.
In this domain, the initial focus is on demand-side pricing and spatial-temporal matching to balance overall network capacity, such as using surge pricing to manage customer demand and network efficiency \citep{castillo2017surge,bimpikis2019spatial}.
On the supply side, existing research primarily evaluates optimal wage contracts and commission structures under the assumption of aggregate supply responses or self-scheduling capacity \citep{cachon2017role,taylor2018demand}.
Building on this, recent literature has increasingly examined driver pricing and scheduling through a macroscopic, equilibrium lens \citep{besbes2021surge,ma2022spatio,garg2022driver}.
These approaches generally assume the platform acts as a macroscopic market maker, optimizing market-clearing prices across a broader pool of agents rather than executing sequential, targeted offers to specific heterogeneous workers.
%


\paragraph{Sequencing and sequential pricing.}
The stochastic probing literature largely focuses on algorithmic approximations for environments with exogenous acceptance probabilities. 
\citet{purohit2019hiring} formalize the hiring under uncertainty problem, demonstrating that greedy heuristics are generally suboptimal under arbitrary acceptance probabilities. 
\citet{epstein2024selection} employ linear programming relaxations and dependent rounding to establish constant-factor approximations for sequential and parallel offering models subject to strict capacity constraints. 
%
%
A fundamental distinction in our setting is that acceptance probabilities are endogenous via the optimal wage decision.
Surprisingly, this can make the optimal policy easier to compute by preventing adversarial sequence of acceptance probabilities. 

A second stream of literature scales sequential decision-making to endogenous settings, where platform actions jointly dictate acceptance probabilities and rewards. Performance in this domain is typically evaluated against offline prophet benchmarks. While allowing the decision-maker to optimize the sequence---the \textit{free-order} variant---improves competitive ratios \citep{beyhaghi2021improved,peng2022order}, the analytical focus remains reliant on fractional approximations. Extending these frameworks to two-sided networks, \citet{pollner2024improved} analyze sequential pricing on bipartite graphs. More recently, \citet{derakhshan2026approximation} investigate an action-reward framework, utilizing a configuration linear program and an efficient polynomial-time approximation scheme (EPTAS) oracle to bypass severe integrality gaps. Our work departs from this reliance on algorithmic approximation; we return to exact structural optimality, mapping the specific distributional conditions under which exact dynamic programming aligns perfectly with simple greedy sequencing.


\section{The Model}
\label{sec:model}

We consider the sequential offering problem faced by on-demand platforms. 
The platform seeks to fulfill a job, which yields a value $v$ if completed, by matching it with one worker from a finite pool $I=[n]$.
To fulfill the job, worker $i \in I$ incurs a total cost of $c_i + Z_i$, where $c_i$ is a deterministic scalar that captures platform-observable components, and $Z_i$ is a random variable with cumulative distribution function (CDF) $F$ representing private factors.
For instance, in food delivery platforms, $c_i$ can account for attributes such as heterogeneous locations and expected travel times, as well as batching efficiencies if the worker already has another job in their bag.
Within this same context, $Z_i$ captures unobservable idiosyncratic preferences that are treated as independent and identically distributed (i.i.d.) across workers to comply with fairness norms and regulations against price discrimination.
Given this cost structure, worker $i$ will accept a \textit{total} wage of $w^{\mathsf{total}}_i$ if and only if it compensates for their total cost, $w^{\mathsf{total}}_i \ge c_i + Z_i$. 
Upon acceptance, the platform realizes a profit of $v - w^{\mathsf{total}}_i$, and the worker derives a surplus of $w^{\mathsf{total}}_i - c_i - Z_i$. 
The total welfare generated by the match is the sum of these payoffs, $v - c_i - Z_i$.
We normalize the platform's offer and the job's value against the observable cost $c_i$ by defining the \textit{net wage} as $w_i \defeq w^{\mathsf{total}}_i - c_i$ and the \textit{net value} as $v_i \defeq v - c_i$. 
Under these definitions, the worker accepts the offer when their private cost falls below the net wage ($Z_i \le w_i$), yielding an acceptance probability of $F(w_i)$.  If they accept, the platform's profit is $v_i - w_i$ and the total welfare generated is $v_i - Z_i$. 
Thus, the problem environment is fully characterized by the distribution $F$ and net values $\vv=(v_1,\ldots,v_n)\in\R^n$.
Throughout, we assume $\E[Z^-]<\infty$, where $Z^-\defeq\max\{-Z,0\}$, so expected welfare is finite for every finite worker pool. 
The platform's problem consists of making sequential take-it-or-leave-it offers, jointly deciding the sequence in which workers are approached and the net wage to offer. 
Formally, the platform selects a \textit{ranking} policy $\sigma$, defined as a bijection $\sigma:[n]\to I$, that dictates the order in which workers are approached, alongside a sequence of net wages $\vw = (w_1, w_2, \dots, w_n) \in \mathbb{R}^n$.
At each step $t = 1, \dots, n$, the platform approaches worker $\sigma(t)$ and offers $w_t$. 
With probability $F(w_t)$, the offer is accepted, and the job is fulfilled, yielding welfare of $v_{\sigma(t)} - Z_{\sigma(t)}$.\footnote{While our main exposition focuses on welfare, in \Cref{sec:quantile_framework} we generalize our framework to capture profit maximization.} 
If the offer is rejected, the platform proceeds to step $t+1$. 
Since private costs $(Z_i)_i$ are i.i.d., the platform's objective depends on the ranking policy solely through the induced sequence of net values.
Let $\Pi(\vv) \defeq \{ (v_{\sigma(1)}, v_{\sigma(2)}, \dots, v_{\sigma(n)}) \mid \sigma : [n] \to I \text{ is bijective} \}$ denote the set of all such permutations. 
We streamline our notation by evaluating policies directly via their corresponding sequence $\vu\in \Pi(\vv)$, expressing the expected welfare under net wages $\vw$ as
\begin{equation} \label{eq:welfare_def}
V(\vu, \vw) = \sum_{t=1}^n \left( \prod_{s=1}^{t-1} (1 - F(w_s)) \right) F(w_t)\,\E[u_t - Z \mid  Z \le w_t].
\end{equation}
To evaluate and compare ranking policies under their respective optimal wages, we overload the notation and define the maximum expected welfare achievable under sequence $\vu$ as $V(\vu)$, and the performance of the optimal policy as $\opt(\vv)\defeq \max_{\vu\in\Pi(\vv)}V(\vu)$.
We assume without loss of generality that the worker pool is indexed such that $v_1 \ge v_2 \ge \dots \ge v_n$, such that the \textit{greedy ranking policy} is represented by $\vv \in \Pi(\vv)$, and we denote its performance by $\gre(\vv)\defeq V(\vv)$.

This sequential offering problem fits into the broad literature on sequential posted pricing and optimal stopping problems.
In particular, for any given sequence, the optimal wages can be solved efficiently via dynamic programming.
The sequencing decision, however, poses a significant theoretical challenge: although irrelevant in a purely i.i.d. environment, it is computationally NP-hard when arbitrarily heterogeneous distributions are allowed \citep{agrawal2020optimal}.
Crucially, our setting exhibits a \textit{location-shift structure}.
Because workers share the same private cost distribution, their realized welfare values $v_i-Z_i$ are location shifts of a common random variable $-Z$.
This specific structure interpolates between the two contrasting regimes and motivates a natural heuristic: a greedy ranking policy that approaches workers in descending order of their observable net values.

In what follows, we leverage this structure to characterize the optimal wages and, more importantly, to derive a tractable recursion for $V(\vu)$ that allows us to systematically analyze and compare different ranking policies.

\section{Optimality of Greedy Ranking}
\label{sec:optimality}

In this section, we establish the exact optimality of the greedy ranking policy.
We first derive a closed-form recursive solution for the optimal wages and welfare objective under an arbitrary sequence, and then leverage the structure to study adjacent pairwise swaps.


\subsection{Optimal Wages}
Given a sequence $\vu\in\Pi(\vv)$, the wage optimization problem is a finite-horizon dynamic program over the sequence of workers. 
Let $V(\vu_{t:n})$ denote the optimal expected continuation value derived from the remaining sequence of workers from step $t$ to $n$, with the convention that $V(\vu_{n+1:n})=0$.
At each step $t \in [n]$, the platform decides the wage $w_t$ to offer the current worker, balancing the immediate expected welfare against the continuation value. 
The Bellman equation is:
\begin{align}
V(\vu_{t:n})
&=
\max_{w_t\in \R}\left\{
F(w_t)\big(u_t-\E[Z\mid Z\le w_t]\big)
+\big(1-F(w_t)\big)V(\vu_{t+1:n})
\right\} \notag\\
&=
V(\vu_{t+1:n})
+
\max_{w_t\in \R}\left\{
F(w_t)\big(u_t-\E[Z\mid Z\le w_t]-V(\vu_{t+1:n})\big)
\right\}.
\label{eq:bellman}
\end{align}

Notice that the optimal wage depends entirely on the marginal value of the current worker on top of the downstream continuation value, $x \defeq u_t - V(\vu_{t+1:n})$.
To capture the expected incremental welfare generated at any given step, we define the \textit{incremental value function}
\begin{equation}
    \psi(x) \defeq \max_{w \in \R} \left\{ F(w)\big(x - \E[Z \mid Z \le w]\big) \right\}. \label{eq:incremental_value}
\end{equation}
The function $\psi(x)$ is the option value of approaching a worker whose match value exceeds the downstream continuation value by $x$: only realizations with private cost below $x$ add welfare.
Its derivative, $\psi'(x)=F(x)$ at continuity points, is the acceptance probability and hence the marginal welfare return to improving that worker's match value.
In particular, $\psi$ is non-decreasing an 1-Lipschitz.
The following lemma establishes that it is welfare-optimal to set the net wage equal to the marginal value, causing worker $\sigma(t)$ to accept whenever $V(\vu_{t+1:n})\le u_t - Z$, i.e.\ whenever the welfare from them accepting exceeds the value-to-go $V(\vu_{t+1:n})$.  

\begin{lemma} \label{lem:optimal_wages}
For any marginal value $x\in\R$, the maximization in the incremental value function is solved by setting $w = x$, which evaluates to:
\begin{equation}
    \psi(x) = \E[\max\{0, x - Z\}] = \int_{-\infty}^{x} F(s)ds. \label{eq:psi_integral}
\end{equation}
Consequently, the optimal wage at step $t$ is $w_t = u_t - V(\vu_{t+1:n})$, and the optimal continuation values satisfy the recursion:
\begin{align}
    V(\vu_{t:n}) &= V(\vu_{t+1:n}) + \psi\big(u_t - V(\vu_{t+1:n})\big), \quad t \in [n]. \label{eq:value_update}
\end{align}
\end{lemma}

\proof{Proof.}
Let $x\in\R$.
%
Notice that for any wage $w\in\R$, we have $F(w)\E[Z \mid Z \le w] = \E[Z \indicator\{Z \le w\}]$.
Hence, the maximization objective in \eqref{eq:incremental_value} simplifies to $\max_{w \in \R} \E[(x - Z)\indicator\{Z \le w\}]$, which is at most $\E[\max\{0, x-Z\}]$.
%
%
This upper bound is achieved by setting $w = x$, which aligns the indicator to include all and only the realizations where $x-Z \ge 0$. 
%
%
Finally, by expressing $\max\{0, x-Z\} = \int_{-\infty}^x \indicator\{Z \le s\} ds$, we can apply Tonelli's Theorem to interchange the expectation and integral, completing the proof.
\hfill\Halmos
\endproof

As a consequence of \Cref{lem:optimal_wages}, the optimal total wage $w^{\mathrm{total}}_t=c_{\sigma(t)}+w_t=v-V(u_{t+1:n})$ weakly increases after each rejection because the continuation value decreases.

\subsection{Local Stability}
\label{subsec:local_optimality}

Before proving global optimality, we establish that when wages are correctly optimized for the greedy sequence $\vv\in\Pi(\vv)$, no other ranking achieves higher expected welfare at those specific wages. 
This ensures that if the current wages are optimized for the greedy sequence, then updating the sequence alone cannot improve the welfare.
This result acts as an essential component of our global optimality proof.

To formalize this, consider two adjacent workers in a greedy sequence at positions $t$ and $t+1$. 
A sufficient condition for stability against a pairwise swap is that, conditional on reaching this pair, the probability that the downstream worker fulfills the job does not exceed that of the upstream worker.
If the probability inequality is reversed, a swap strictly improves welfare whenever $v_t>v_{t+1}$ and the pair is reached with positive probability. 
Equal values produce no gain from swapping at fixed wages.
Therefore, stability requires:
\begin{equation}
    F(w_t) \ge F(w_{t+1})\big(1 - F(w_t)\big), \quad \forall\,t\in[n-1]. \label{eq:local_stability_prob}
\end{equation}
Let $x=v_{t+1}-V(\vv_{t+2:n})$ and $y=v_{t}-V(\vv_{t+2:n})$ denote their respective marginal values over the shared future continuation value.
By \Cref{lem:optimal_wages}, it is welfare-optimal to set $w_{t+1} = x$ and $w_t = y - \psi(x)$. 
Substituting these into the stability condition and rearranging gives $F(y-\psi(x)) \ge \frac{F(x)}{1+F(x)}$.
Because $y \ge x$ and $F$ is non-decreasing, the left-hand side is naturally bounded below by $F(x-\psi(x))$. 
Thus, to guarantee local stability for all possible value pairs, it is sufficient that the distribution satisfies:
\begin{equation}\label{eq:local_stability_under_optwages}
    F(x-\psi(x)) \ge \frac{F(x)}{1+F(x)}, \quad \forall\, x \in \R.
\end{equation}
The following lemma establishes that this property is satisfied when the private cost CDF $F$ is concave, representing diminishing marginal returns in acceptance probability.
This requires the support of $F$ to be bounded below, which we normalize to $0$ without loss of generality.

\begin{lemma} \label{lem:local_stability}
Suppose that the lower endpoint of the support of $F$ is $0$, possibly with an atom $p_0=\P(Z=0)$, and that it is concave on the convex hull of its support.
Then, \eqref{eq:local_stability_under_optwages} holds for every $x\in\R$ and greedy ranking is stable under welfare-optimal wages.
\end{lemma}

\proof{Proof.}
For $x< 0$, both sides of \eqref{eq:local_stability_under_optwages} are zero.
At $x=0$, the inequality is $p_0\ge p_0/(1+p_0)$.
Fix $x>0$ and let $Z_x = \min\{Z, x\}$.
From \eqref{eq:psi_integral} in \Cref{lem:optimal_wages} we get that
\[ x-\psi(x) = x-\E[\max\{0,x-Z\}] = \E[Z_x]. \]
Concavity of $F$ on the support of $Z_x$ and Jensen's inequality imply $F(x-\psi(x)) \ge \E[F(Z_x)]$. 
Evaluating the right-hand side and accounting for the possible jump $p_0$ at the lower endpoint gives
\begin{align*}
    \E[F(Z_x)] 
    &= p_0^2 + \int_{(0,x]} F(t)dF(t) + (1 - F(x))F(x) \\
    &= F(x) - \frac{1}{2}F(x)^2 + \frac{1}{2}p_0^2.
\end{align*}
The last expression is at least $F(x)/(1+F(x))$, because 
\begin{align*}
    \left(F-\frac{F^2}{2}\right)(1+F) - F = \frac{F^2(1-F)}{2}\ge0,
\end{align*}
and the atom contributes an additional non-negative term.
If the upper support endpoint is finite, extend $F$ to be constant equal to one beyond it and the argument follows since concavity is preserved.
Finally, given that \eqref{eq:local_stability_under_optwages} (and thus \eqref{eq:local_stability_prob}) holds, the overall probability weight $a_t = F(w_t) \prod_{s=1}^{t-1} (1 - F(w_s))$ in \eqref{eq:welfare_def} is decreasing in $t$ and thus pairing them with decreasing values maximizes welfare over all permutations.
This completes the proof.
\hfill\Halmos
\endproof

\subsection{Global Optimality}
\label{subsec:global_optimality}
In the previous subsection, we established that the greedy ranking is locally stable when paired with its optimal wages, provided the cost distribution features diminishing marginal returns.
However, local stability between two workers does not automatically guarantee that the greedy assignment is globally optimal against all possible permutations.

To prove global optimality, we use a swapping argument. 
Consider an arbitrary non-greedy sequence $\vu \in \Pi(\vv)\setminus\{\vv\}$ and identify an adjacent inverted pair, where the worker at the upstream position has a strictly lower observable value than the worker at the downstream position. 
We then evaluate the impact of swapping their positions to place the superior worker first. 
Holding the wages fixed does not establish the desired comparison: a different sequencing must be evaluated under its own optimal wages.
Changing the sequence changes the continuation values and hence optimal wages.
Therefore, we must dynamically recalculate the optimal prices for the newly swapped pair.

Let $t$ and $t+1$ be the positions of an adjacent inverted pair, such that $u_t < u_{t+1}$. We define the marginal values of these two workers over the shared future continuation value as $x = u_t - V(\vu_{t+2:n})$ and $y = u_{t+1} - V(\vu_{t+2:n})$. Because the pair is inverted, we have $x < y$.
Applying the recursion from \Cref{lem:optimal_wages}, we directly evaluate the expected continuation value extracted from this pair in the original non-greedy order:
\begin{align*}
    V(\vu_{t+1:n}) &= V(\vu_{t+2:n}) + \psi(u_{t+1} - V(\vu_{t+2:n})) \\
    &= V(\vu_{t+2:n}) + \psi(y), \text{ and} \\ 
    V(\vu_{t:n})   &= V(\vu_{t+1:n}) + \psi\big(u_t - V(\vu_{t+1:n})\big) \\
                   &= V(\vu_{t+2:n}) + \psi(y) + \psi\big(x - \psi(y)\big).
\end{align*}


If we swap their positions and re-optimize the wages, the downstream worker now has margin $x$, adding $\psi(x)$ to the continuation value $V(\vu_{t+2:n})$. The upstream worker faces a margin of $y - \psi(x)$, resulting in a new expected continuation value of $V(\vu_{t+2:n}) + \psi(x) + \psi(y - \psi(x))$. 

We define the swap difference $D(x,y)$ as the expected welfare difference between the swapped sequence and the original sequence:
\begin{equation}
    D(x,y) \defeq \psi\big(y - \psi(x)\big) + \psi(x) - \psi\big(x - \psi(y)\big) - \psi(y). \label{eq:swap_difference}
\end{equation}
If $D(x,y) \ge 0$ for all $x \le y$, then this swap weakly improves expected welfare. 
The update $z\mapsto z + \psi(u-z)$ is non-decreasing because $\psi$ is 1-Lipschitz.
Thus, an improvement at the swapped pair propagates through all preceding offers when their wages are re-optimized.
Iteratively removing adjacent inversions transforms any sequence into $\vv$ without lowering expected welfare. 
This establishes that no sequence strictly outperforms the greedy sequence.

To guarantee that this swap difference is always non-negative, we impose an additional assumption on the distribution. 
While local stability in \Cref{lem:local_stability} is satisfied by a non-increasing density, global optimality is guaranteed if we further assume this density is convex.

\begin{theorem} \label{thm:global_optimality}
Suppose that the lower endpoint of the support of $F$ is 0, possibly with an atom $p_0=\P(Z=0)$. 
On the open support $(0,b)$, where $b\in(0,\infty]$, suppose that $F$ is atomless and has a continuously differentiable density $f$ that is non-increasing and convex. 
If $b<\infty$, assume that there is no atom at $b$ and allow the density to make one terminal jump to zero there.
Then, $\gre(\vv)=\opt(\vv)$ for all $\vv\in\R^n$.
\end{theorem}

\proof{Proof.}
Fix $x\le y$. 
We first note that $\psi(x)= 0$ for $x\le 0$, since $z\ge 0$.
Thus, if $y\le 0$, then $D(x,y)=0$. If $y>0$ and $x\le\psi(y)$, then $x-\psi(y)\le 0$ and we claim that
\begin{align}\label{eq:trivial_case}
    D(x,y) = \psi(y-\psi(x))+\psi(x)-\psi(y)\ge 0.
\end{align}

Indeed, $\psi(x-\psi(y))=0$ when $x-\psi(y)\le 0$, because $\psi$ is non-decreasing and 1-Lipschitz we get
\begin{align*}
    \psi(y)-\psi(y-\psi(x))=\int_{y-\psi(x)}^y\psi'(t)\,dt\le\int_{y-\psi(x)}^y1\,dt = \psi(x).
\end{align*} 
In particular, $D(\psi(y),y)\ge 0$.

It remains to consider $\psi(y)<x<y$. First, suppose either $b=\infty$ or $y<b$, in which case all arguments below lie in the smooth part of the support.
%
%
We evaluate the first partial derivative with respect to $x$:
\begin{align}
    \frac{\partial D}{\partial x}(x,y) = \psi'(x)\big[1 - \psi'\big(y - \psi(x)\big)\big] - \psi'\big(x - \psi(y)\big). \label{eq:first_derivative}
\end{align}

At $x=y$, \Cref{lem:local_stability} gives $\frac{\partial D}{\partial x}(y,y) = \psi'(y)[1 - \psi'(y - \psi(y))] - \psi'(y - \psi(y)) \le 0$.
Two further differentiations yield
\begin{align}
    \frac{\partial^3 D}{\partial x^3}(x,y) 
    = &f'(x)\big[1 - F\big(y - \psi(x)\big)\big] \label{eq:third_derivative_1}\\ 
    &+ 3F(x)f(x)f\big(y - \psi(x)\big)\label{eq:third_derivative_2} \\ 
    &- F(x)^3 f'\big(y - \psi(x)\big) \label{eq:third_derivative_3}\\ 
    &- f'\big(x - \psi(y)\big). \label{eq:third_derivative_4}
\end{align}
The middle two terms \eqref{eq:third_derivative_2} and \eqref{eq:third_derivative_3} are non-negative because $f'\le 0$.
The remaining two can be written as 
\begin{align*}
    f'(x)-f'(x-\psi(y)) - f'(x)F(y-\psi(x)),
\end{align*}
which is also non-negative because $f'$ is non-decreasing.
Hence, the partial derivative $\frac{\partial D}{\partial x}(x,y)$ is a convex function of $x$.
A convex derivative that ends at a non-positive value at $x = y$ can change sign from positive to negative at most once. 
Consequently $D(\cdot,y)$ is quasi-concave on $[\psi(y),y]$.
Using both endpoint values, $D(\psi(y),y)\ge0$ from \eqref{eq:trivial_case} and $D(y,y)=0$, yields $D(x,y)\ge0$ throughout this interval.

For completeness, suppose $b<\infty$ and $y\ge b$, and let $\mu=\E[Z]$ and set $c=y-x$. 
The case $x\le\psi(y)=y-\mu$ was already handled, so take $0\le c<\mu$. 
%
%
Since $Z\le b\le y$, we have $c \ge \E[(Z-x)^+] = \psi(x) - x + \mu$.
With $\delta=c-\psi(x) + x - \mu$, direct substitution gives
\[
D(x,y)=\psi(\mu+\delta)-\psi(\mu-c)-\delta.
\]
For fixed $c$, the right-hand side is non-increasing in $\delta\in[0,c]$, so it is bounded below by $\kappa(c)\defeq\psi(\mu+c)-\psi(\mu-c)-c$. 
Moreover, $\kappa(0)=0$, $\kappa(\mu)=\E[(Z-2\mu)^+]\ge0$, and, wherever differentiated,
\[
\kappa''(c)=f(\mu+c)-f(\mu-c)\le0.
\]
Thus $\kappa$ is concave on $[0,\mu]$, including across the permitted terminal density drop, and lies above the chord joining its non-negative endpoint values. 
Hence $D(x,y)\ge0$ also in the finite-support boundary region.

Every adjacent inversion can therefore be swapped and the two wages re-optimized without lowering welfare. 
Iterating these swaps produces the greedy sequence, which proves the result.
\hfill\Halmos
\endproof

Theorem \ref{thm:global_optimality} establishes that the greedy sequence is globally optimal when the cost distribution exhibits a non-increasing, convex density. 
This structural condition corresponds to environments where wage incentives yield diminishing marginal returns on acceptance probabilities at a decelerating rate.
Intuitively, a decreasing density makes selective early offers effective: lowering an offer to account for continuation value retains enough acceptance probability to favor better matches.
Convexity adds a restriction on how fast that sensitivity can change.
Thus, the two shape assumptions control both the immediate allocation probabilities and the effect of re-optimizing offers.
While these properties are satisfied by several standard distributions, they are sufficient rather than strictly necessary.
%
%
The next section gives a distribution-free performance guarantee when these conditions are not satisfied.

\section{Worst-case Guarantees}
\label{sec:worst_case_guarantees}

We benchmark the greedy ranking against two baselines: the online optimum $\opt(\vv) = \max_{\vu \in \Pi(\vv)} V(\vu)$, which optimizes over all sequences, and the offline prophet $\off(\vv) \defeq \E[\max_{i \in I}\{(v_i - Z_i)^+\}]$, which observes all cost realizations in hindsight.
The classical prophet inequality establishes that even if we cannot choose the order, there is a universal lower bound of $1/2$: for any distribution $F$, net values $\vv$, and sequence length $n$, an optimal stopping policy for any fixed order achieves at least half the prophet benchmark \citep{krengel1977semiamarts, samuel1984comparison}.

This bound is tight in the general heterogeneous setting, where the classic worst-case instance constructs variables with vastly different variances. 
However, this instance cannot be cast under our location-shift structure, $X_i = (v_i - Z_i)^+$ with $Z_i \sim F$, because the untruncated welfare variables share a common additive noise distribution. 
Since the greedy ranking evaluates workers in decreasing order of observable values ($v_1 \ge \dots \ge v_n$), the location-shift structure guarantees that the expected positive surplus from the current worker bounds that of any subsequent worker. 
By exploiting this monotonicity, we prove a tighter lower bound for finite sequences. Specifically, our bound is parameterized by $n$ and performs better for small $n$, even though it also converges to the standard $1/2$ in the worst case as $n\to\infty$.

\begin{theorem}\label{thm:worst_case}
For any $n \ge 1$, net values $v_1 \ge v_2 \ge \dots \ge v_n$, and cost distribution $F$ with $\E[Z^-]<\infty$,
\begin{equation}
V(\vv) \ge \frac{n}{2n-1}\off(\vv).
\end{equation}
\end{theorem}

\proof{Proof.}
For $n=1$, the equality holds by \Cref{lem:optimal_wages}.
Suppose $n\ge 2$.
Let $G_k \defeq V(\vv_{k:n})$ denote the greedy continuation value from the $k$-th worker onwards, with $G_{n+1} = 0$. By \Cref{lem:optimal_wages}, $G_k = G_{k+1} + \psi(v_k - G_{k+1})$ for each $k \in [n]$, and $G_1 = V(\vv)$.

For any constant $c \ge 0$ and random variables $X_i = (v_i - Z_i)^+$, the pointwise bound $\max_i X_i \le c + \sum_i (X_i - c)^+$ holds. Setting $c = G_2 \ge 0$ and taking expectations, the identity $\E[((v_i - Z_i)^+ - c)^+] = \psi(v_i - c)$ yields
\begin{equation*}
\off(\vv) \le G_2 + \psi(v_1 - G_2) + \sum_{j=2}^n \psi(v_j - G_2).
\end{equation*}
The first two terms equal $G_1$ by the recursion for $G_1$. Defining $S_2 \defeq \sum_{j=2}^n \psi(v_j - G_2)$,
\begin{equation*}
\off(\vv) \le G_1 + S_2.
\end{equation*}

We derive two upper bounds on $S_2$. First, since $v_1 \ge v_j$ for all $j \ge 2$ and $\psi$ is non-decreasing,
\begin{equation*}
S_2 \le (n-1)\,\psi(v_1 - G_2).
\end{equation*}
Second, define $S_k = \sum_{j=k}^n \psi(v_j - G_k)$ for $k \ge 2$, with $S_{n+1} = 0$. Since $G_k \ge G_{k+1}$, monotonicity of $\psi$ gives $\psi(v_j - G_k) \le \psi(v_j - G_{k+1})$ for $j \ge k+1$. The same monotonicity also bounds the leading term, so
\begin{equation*}
S_k \le \psi(v_k-G_{k+1})+S_{k+1}
     =(G_k-G_{k+1})+S_{k+1}.
\end{equation*}
Telescoping from $k=2$ to $n$ yields $S_2 \le G_2$.

Combining $(n-1)$ copies of $S_2 \le G_2$ with one copy of $S_2 \le (n-1)\,\psi(v_1 - G_2)$,
\begin{equation*}
n \cdot S_2 \le (n-1)\big[G_2 + \psi(v_1 - G_2)\big] = (n-1)\,G_1,
\end{equation*}
so $S_2 \le \frac{n-1}{n}\,G_1$. Substituting back,
\begin{equation*}
\off(\vv) \le G_1 + \tfrac{n-1}{n}\,G_1 = \tfrac{2n-1}{n}\,G_1,
\end{equation*}
completing the proof.
\hfill\Halmos
\endproof



As $n \to \infty$, the lower bound $n/(2n-1)$ converges to $1/2$. 
To establish that this finite-horizon bound is tight, we evaluate greedy ranking and the optimal online policy $\opt(\vv)$ under a Bernoulli distribution.

For $n\ge2$, let $q=1-p\in(0,1/n]$, and let $Z\in\{0,1\}$ with $\P(Z=1)=p$. Set
\[
v_k=1+(n-k)q,\qquad k\in[n].
\]
Here $\psi(w)=qw$ on $[0,1]$ and $\psi(w)=w-p$ for $w\ge1$. Backward induction gives $G_k=(n-k+1)q$: every greedy margin is exactly one. Thus $\gre(\vv)=nq$.

For comparison, consider the feasible order $(v_2,v_3,\ldots,v_n,v_1)$. Its last-stage value is $\psi(v_1)=nq$. At its preceding position $j\in[n-1]$, backward induction gives the wage
\[
w_j=1-nq p^{\,n-j-1}\in[0,1]
\]
and continuation value $(n-j)q+nq p^{\,n-j}$. In particular, this order has value
\[
H=(n-1)q+nq p^{n-1}\le\opt(\vv)\le\off(\vv).
\]
Applying \Cref{thm:worst_case} and the feasible-policy lower bound gives
\[
\frac{n}{2n-1}
\le\frac{\gre(\vv)}{\off(\vv)}
\le\frac{\gre(\vv)}{\opt(\vv)}
\le\frac{\gre(\vv)}{H}
=\frac{n}{(n-1)+np^{n-1}}.
\]
As $p\uparrow1$, the rightmost expression converges to $n/(2n-1)$. Both ratios therefore have this limit, proving tightness against both benchmarks for every $n\ge2$. For $n=1$, greedy equals both benchmarks.

In the following section, we explore the empirical robustness of the greedy heuristic when the assumptions in \Cref{thm:global_optimality} are not satisfied.

\section{Numerical Results}
\label{sec:simulations}

In \Cref{sec:optimality}, we establish that greedy ranking is welfare-optimal when the private cost distribution exhibits a non-increasing and convex probability density function.
%
%
In this section, we study the empirical performance of the greedy ranking policy across a broad class of distribution families that do not satisfy our sufficient conditions from \Cref{thm:global_optimality}.
We find no counterexample for several such families.
We then turn to three families for which greedy can be suboptimal and search for instances with larger welfare losses.
Even in these deliberately adverse cases, the lowest performance ratios we found remain well above the universal guarantee in \Cref{sec:worst_case_guarantees}.
Throughout this section, every ranking uses its own welfare-optimal wages.

\subsection{Simulation setup}

A \textit{counterexample} to greedy optimality is a set of worker values for which another ranking yields higher expected welfare.
The swap argument in \Cref{subsec:global_optimality} reduces the question to two workers: for values $x<y$, the swap difference $D(x,y)$ from \eqref{eq:swap_difference} is the welfare of the greedy sequence $(y,x)$ minus the welfare of the reverse sequence $(x,y)$.
A negative value is a counterexample.
We search across the center and tails of each tested distribution, using a fine grid and optimization between grid points, and recheck the most adverse candidates at high precision.
We find no counterexample for Normal, Log-normal, Log-logistic, or Gumbel (for maxima) costs within the domains in \Cref{tab:distribution_robustness}. 
Together, these searches evaluate more than five million grid points. \Cref{sec:search_protocol} gives the complete design and accuracy checks.

\begin{table}[H]
\centering
\small
\begin{tabular}{@{}lll@{}}
\hline
\textbf{Cost family} & \textbf{Shape range searched} & \textbf{Result} \\
\hline
Normal & No shape parameter & No counterexample found \\
Log-normal & $0.1\le\sigma\le3$ & No counterexample found \\
Log-logistic & $1.01\le k\le20$ & No counterexample found \\
Gumbel for maxima & No shape parameter & No counterexample found \\
\hline
\end{tabular}
\caption{Families with no counterexample found. The search covers probabilities from $10^{-12}$ to $1-10^{-12}$ and, for nonnegative costs, the full remaining band of lower worker values after excluding a region where greedy is provably better. Location and scale changes are handled analytically; only the displayed shape parameters require separate searches.}
\label{tab:distribution_robustness}
\end{table}

\subsection{When Greedy is Suboptimal}

Counterexamples can be constructed whenever the mean cost lies below the median. 
For a continuous distribution with finite mean $\mu$, the condition (proved in \Cref{lem:mean-median-obstruction}) is
\begin{align}
F(\mu)<\frac12
\quad\Longrightarrow\quad
\text{greedy is strictly suboptimal for some two-worker instance.}
\label{eq:failure-condition-main}
\end{align}
The construction uses two high, closely spaced match values. Under greedy ranking, the first wage is then close to the mean cost and is accepted with probability below one half. 
Approaching the lower-value worker first can improve welfare by preserving the better match as a fallback. 
Three families illustrate the condition across different cost supports.
Beta$(\alpha,1)$ has compact support $[0,1]$ and density $f(z)=\alpha z^{\alpha-1}$. 
For $0<\alpha\le1$, the density is non-increasing and convex, so \Cref{thm:global_optimality} applies directly.
For every $\alpha>1$, the mean lies below the median and counterexamples exist. 
Thus this family has a sharp transition at $\alpha=1$ (see \Cref{prop:beta-alpha-one-phase}).
Weibull costs have support $[0,\infty)$ and CDF $F_k(z)=1-e^{-z^k}$. 
Greedy is optimal for $k\le1$ by \Cref{thm:global_optimality}. The mean falls below the median when $k>k_0$, where
\[
\bigl[\Gamma(1+1/k_0)\bigr]^{k_0}=\log 2,
\qquad k_0\approx3.43954.
\]
Our search finds no counterexample for the tested shapes between $1$ and $k_0$, where optimality remains unresolved.

Gumbel for minima has support $\mathbb R$ and CDF $F(z)=1-e^{-e^z}$. Its mean $-\gamma_{\rm E}$ is below its median $\log(\log 2)$, so counterexamples also exist. 
This contrasts the search outcome for Gumbel for maxima. 
Moreover, if $Z_k$ is Weibull, then $k(Z_k-1)$ converges to Gumbel for minima. 
This links the two families: Gumbel for minima describes the welfare comparisons that arise as Weibull costs concentrate near one (see \Cref{sec:failure_proofs,prop:weibull-gumbel-limit}).

\subsection{How Much Welfare Does Greedy Lose?}

For each tested value vector $\vv$, we compute the ratio $\gre(\vv)/\opt(\vv)$, where the online optimum can choose any ranking and optimize its wages.
We search for low ratios, computing the best ranking for each candidate dynamic programming over all subsets of workers.
Details for the numerical search over value vector appear in \Cref{sec:performance_protocol}.
With two workers, the ratio gradually decreases for parameters beyond those that break the mean-median condition in \eqref{eq:failure-condition-main}.
For Weibull, the lowest ratio found declines from $99.9\%$ at $k=3.5$ to $97.9\%$ at $k=10$ and $96.5\%$ at $k=100$, close to the $96.3\%$ found for Gumbel (for minima). 
For Beta, it falls to $96.8\%$ at $\alpha=2$ and $91.6\%$ at $\alpha=1000$.

\begin{figure}[]

\centering
\begin{subfigure}[t]{0.49\textwidth}
\centering
\includegraphics[width=\linewidth]{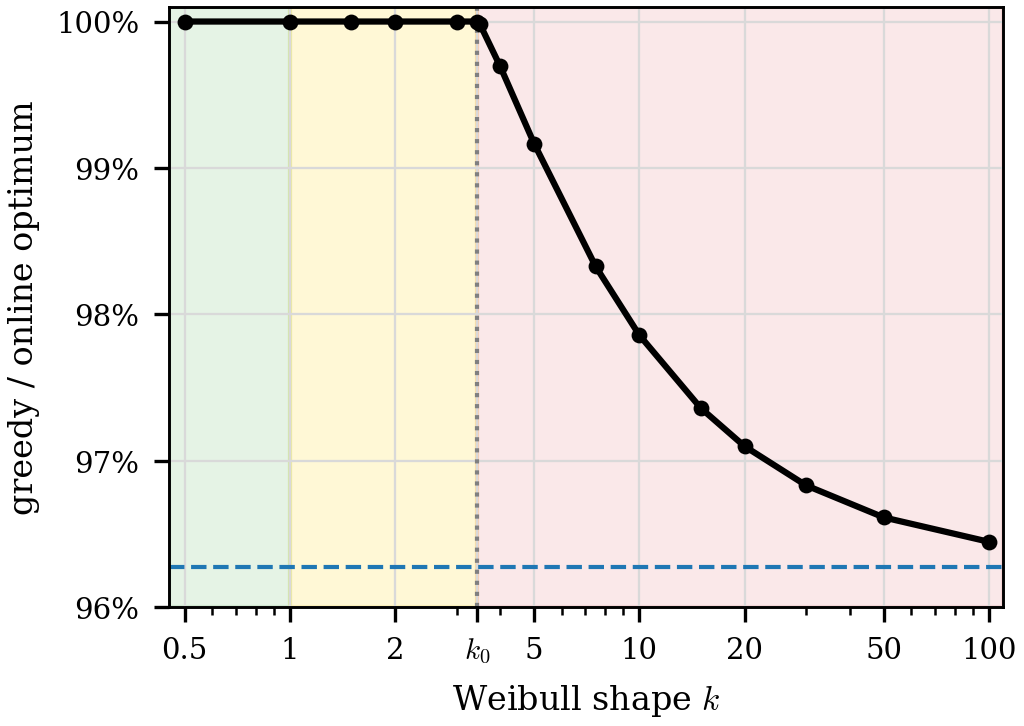}
\caption{Two workers: increasing Weibull shape.}
\label{fig:weibull_by_shape}
\end{subfigure}
\hfill
\begin{subfigure}[t]{0.49\textwidth}
\centering
\includegraphics[width=\linewidth]{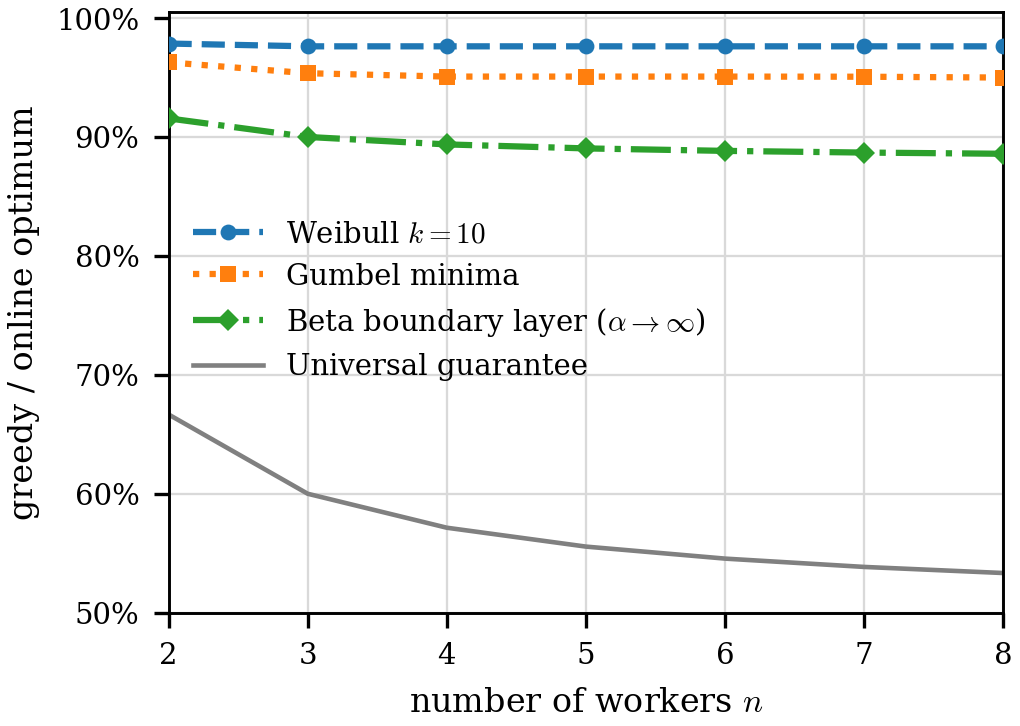}
\caption{Increasing the number of workers.}
\label{fig:failure_by_horizon}
\end{subfigure}
\caption{Lowest welfare ratios found when greedy can be suboptimal. In panel (a), the shaded regions indicate proved optimality (green), unresolved shapes with no counterexample found (yellow), and shapes satisfying the mean--median obstruction (red); $k_0$ marks the crossing. The dashed line is the lowest two-worker ratio found for Gumbel minima. Panel (b) compares Weibull with $k=10$, Gumbel minima, and the limiting Beta experiment as $\alpha\to\infty$. The solid gray line is the universal guarantee $n/(2n-1)$.}
\label{fig:suboptimal_family_performance}
\end{figure}

Adding workers produces further, but modest, losses. 
At eight workers, the lowest ratios found are $97.6\%$ for Weibull with $k=10$ and $95.0\%$ for Gumbel minima.
Beta produces the largest loss: as costs concentrate at their upper endpoint, the ratio falls from $91.566\%$ with two workers to $88.589\%$ with eight. 
The experiment keeps match values within a shrinking neighborhood of that endpoint, where ranking continues to matter (see \Cref{sec:beta-alpha-one} for additional details);

For comparison, the universal guarantee at eight workers is only $53.3\%$. 
The welfare cost of greedy in these deliberately adverse families is therefore much smaller than the distribution-free worst case suggests. 
It is worth noting that these are the lowest ratios found, not stronger guarantees, and more severe instances may exist. 
Nonetheless, greedy retains most of the optimal welfare even in families selected specifically because it can be suboptimal.

\section{Concluding Remarks and Future Work}

This paper studies the sequential offering problem faced by on-demand platforms, where the sequence of workers and the dynamic wage trajectory must be jointly optimized. 
By identifying a location-shift structure arising from the observable heterogeneity among workers, we bridge the analytical gap between trivial i.i.d. environments and computationally intractable, arbitrarily heterogeneous settings.

Our main result establishes that when the unobservable cost distribution exhibits a non-increasing and convex density, the exact optimal policy is structurally simple: the platform should rank workers greedily by their observable net values and optimize wages dynamically via backward induction. 
In these settings, sending initial low offers to inferior matches cannot improve upon a simple greedy ranking with optimized wages.

The numerical results extend the practical case for greedy ranking.
We find no counterexample for several relevant distributions, and even deliberately adverse instances in families with proved sub-optimality retain much more welfare than the worst-case guarantee.
These results motivate sharper guarantees for specific cost distributions while preserving a simple operational prescription: rank by match value and optimize wages to account for the value of later offers.

Several promising directions remain for future research. 
While our model focuses on sequential search for a single job, extending this framework to incorporate multiple jobs, parallel offerings, or correlated values would capture additional operational complexities. 
Furthermore, integrating online learning into the model to account for shifts in the platform's belief over the cost distribution $F$ presents a valuable avenue for real-world platform operations (see \Cref{sec:decaying_markets} for a natural extension of our result to known deterministic deterioration in cost distributions along the offering sequence).

\clearpage
\bibliographystyle{informs2014} 
\bibliography{DraftBib} 

\ECSwitch



%
%
%

\begin{APPENDICES}
\crefalias{section}{appendix}
\crefalias{subsection}{appendix}
\crefalias{subsubsection}{appendix}

\section{Additional Optimality Results}
\label{sec:additional_optimality}
The two cases below extend optimality beyond the density conditions of \Cref{thm:global_optimality} in different ways. Half-normal costs retain the adjacent-swap proof, with an additional bound replacing density convexity. Logistic costs yield the stronger conclusion that every ranking has the same welfare.

\subsection{Half-normal Costs}
\label{sec:half_normal}
The density of a Half-normal distribution is decreasing but is concave near zero, so it falls outside \Cref{thm:global_optimality}. Nevertheless, the same swap argument applies because its positive derivative terms dominate the potentially negative term directly.

\begin{theorem}\label{thm:half_normal}
If $Z=\mu+\sigma|N|$, where $N$ is standard Normal and $\sigma>0$, greedy ranking with welfare-optimal wages is globally optimal for every finite worker pool and every vector of real net values.
\end{theorem}
\proof{Proof.}
Translation of costs and values by $\mu$ and scaling by $\sigma$ reduce the problem to $Z=|N|$. Put $c_0=\sqrt{2/\pi}$, $f(t)=c_0e^{-t^2/2}$ for $t\ge0$, and $\overline F=1-F$. Then
\[
\psi(t)=tF(t)-c_0(1-e^{-t^2/2}),\qquad t\ge0.
\]
The trivial-region argument in \Cref{thm:global_optimality} handles $x\le\psi(y)$, and \Cref{lem:local_stability} gives $D_x(y,y)\le0$. It remains to prove $D_{xxx}(x,y)\ge0$ when $0<\psi(y)<x\le y$.

Write $a=x-\psi(y)>0$, $b=y-\psi(x)\ge x-\psi(x)>0$, and $h(t)=f(t)/\overline F(t)$. Since $f'(t)=-tf(t)$, differentiation gives
\begin{align*}
D_{xxx}(x,y)
={}&-xf(x)\overline F(b)+3F(x)f(x)f(b)\\
&+F(x)^3bf(b)+af(a).
\end{align*}
We use elementary bounds $3/4<c_0<1$ and $e^{-2}<1/7$. Moreover,
\[
\frac{\overline F(t)}{f(t)}
=\int_0^\infty e^{-tu-u^2/2}\,du\le\frac1{c_0},
\qquad t\ge0,
\]
so $h(t)\ge c_0$. The same integral is at most $1/t$ for $t>0$, giving $\overline F(2)\le f(2)/2<1/14$ and hence $F(2)>13/14$.

If $0<x\le2$, decreasing $f$ implies that $F(x)/x$ is non-increasing, and therefore
\[
3\frac{F(x)}x h(b)\ge3\frac{F(2)}2c_0>
3\frac{13}{28}\frac34=\frac{117}{112}>1.
\]
Thus $3F(x)f(x)f(b)\ge xf(x)\overline F(b)$, and all remaining terms are non-negative.

If $x\ge2$, then
\[
b\ge x-\psi(x)
=x\overline F(x)+c_0(1-e^{-x^2/2})>\frac9{14}.
\]
Also $F(x)>13/14$, and $xf(x)\le2c_0e^{-2}<2/7$, since $te^{-t^2/2}$ decreases for $t\ge1$. Consequently,
\[
F(x)^3b h(b)>
\left(\frac{13}{14}\right)^3\frac9{14}\frac34
>\frac27>xf(x).
\]
Now the third positive term dominates the negative term. This proves $D_{xxx}\ge0$ throughout the nontrivial region.

It follows that $D_x(\cdot,y)$ is convex and ends at a non-positive value. Its non-positive sublevel set is an interval containing $y$, so $D(\cdot,y)$ can first increase and then decrease, or decrease throughout. Its minimum is therefore attained at an endpoint. Since $D(\psi(y),y)\ge0$ and $D(y,y)=0$, every adjacent inversion can be removed without reducing welfare. The monotone Bellman updates propagate this comparison to the full sequence, proving the theorem.
\hfill\Halmos
\endproof

\subsection{Logistic Costs: Every Ranking Has the Same Welfare}
\label{sec:logistic_indifference}

For Logistic costs, ranking is irrelevant once wages are optimized. This follows directly from a closed-form continuation value.

\begin{proposition}\label{prop:logistic_indifference}
If $Z=\mu+sL$, where $L$ is standard Logistic and $s>0$, then for every ranking $\vu$,
\[
V(\vu)=s\log\left(1+\sum_{i=1}^n e^{(u_i-\mu)/s}\right).
\]
Consequently, $D(x,y)=0$ for every pair and all rankings are welfare-equivalent under their optimal wages.
\end{proposition}
\proof{Proof.}
Integrating the CDF gives $\psi(t)=s\log(1+e^{(t-\mu)/s})$. Thus the update in \Cref{lem:optimal_wages} satisfies
\[
\exp\left(\frac{q+\psi(u_i-q)}s\right)
=e^{q/s}+e^{(u_i-\mu)/s}.
\]
Starting with terminal value $q=0$ and working backward yields the displayed formula. It is unchanged by any permutation, proving the claim.
\hfill\Halmos
\endproof

Worker $i$ fulfills the job with probability $e^{(v_i-\mu)/s}/(1+\sum_j e^{(v_j-\mu)/s})$ under the optimal wages. The remaining probability is non-fulfillment. These probabilities have the multinomial-logit form and also do not depend on the order.
%

\section{Searching for Counterexamples to Greedy Optimality}
\label{sec:search_protocol}

This appendix describes the search behind \Cref{sec:simulations}. The task is to find two values $x<y$ for which the lower-value worker should be approached first. We explain the objective, the values and distributions searched, and the checks used to distinguish a counterexample from numerical error. The performance experiments in \Cref{sec:performance_protocol} ask the separate question of how much welfare greedy loses.

\subsection{Objective and Numerical Accuracy}

Recall that $\psi(t)=\E[(t-Z)^+]$ and
\[
D(x,y)=\psi\bigl(y-\psi(x)\bigr)+\psi(x)
       -\psi\bigl(x-\psi(y)\bigr)-\psi(y).
\]
This is greedy welfare minus the welfare of the reverse order. Directly minimizing $D$ is numerically difficult: it is zero when $x=y$, and the four terms can nearly cancel when values are close. We therefore divide by the value difference and minimize $G(x,y)=D(x,y)/(y-x)$. This preserves the sign for $x<y$ and gives the integral representation
\begin{equation}
G(x,y)=\int_0^1\!\left\{
F\bigl(s(u)-\psi(y)\bigr)
-F\bigl(s(u)\bigr)\left[1-F\bigl(y-\psi(s(u))\bigr)\right]
\right\}du,
\quad s(u)=x+u(y-x).
\label{eq:normalized_swap_gap}
\end{equation}
To obtain this identity, integrate $-D_x(s,y)$ from $x$ to $y$ and use $D(y,y)=0$. At equal values, we evaluate its continuous limit,
\[
G(y,y)=F\bigl(y-\psi(y)\bigr)[1+F(y)]-F(y).
\]
A negative limit identifies a promising neighborhood, but every reported counterexample is checked at distinct values.

We approximate the integral by Gauss--Legendre quadrature, a weighted sum of the integrand at selected points. The initial grid uses 32 points. Any value below $-10^{-10}$ is a \emph{screening flag} and is recomputed with both 64 and 96 points. Optimization between grid points always uses the 96-point rule. Finally, the most adverse grid candidates and every optimizer output are checked with 100-decimal-digit arithmetic, using both \eqref{eq:normalized_swap_gap} and the original expression for $D$. A negative initial screen alone is not reported as a counterexample.

\subsection{Distributions and Search Domains}

Location and scale need not be searched separately. For $Z=a+bZ_0$ with $b>0$,
\begin{equation}
\psi_{a,b}(a+bs)=b\psi_0(s),\qquad
D_{a,b}(a+bx,a+by)=bD_0(x,y).
\label{eq:search-location-scale}
\end{equation}
Thus an affine transformation of both costs and worker values preserves the sign of the comparison. We use the standardized laws in \Cref{tab:search-parameter-domains}; $\Phi$ denotes the standard Normal CDF. For nonnegative distributions, the displayed CDF applies for $z>0$ and is zero for $z\le0$.

\begin{table}[H]
\centering
\small
\begin{tabular}{@{}lll@{}}
\hline
\textbf{Family} & \textbf{Standardized CDF} & \textbf{Shape range} \\
\hline
Normal & $\Phi(z)$ & None \\
Log-normal & $\Phi(\log z/\sigma)$ & $\sigma\in[0.1,3]$ \\
Log-logistic & $z^k/(1+z^k)$ & $k\in[1.01,20]$ \\
Weibull & $1-e^{-z^k}$ & $k\in[0.25,20]$ \\
Gumbel for maxima & $e^{-e^{-z}}$ & None \\
Gumbel for minima & $1-e^{-e^z}$ & None \\
\hline
\end{tabular}
\caption{Distribution definitions and finite shape ranges. Each shape range uses 31 geometrically spaced values, together with optimization over the continuous range.}
\label{tab:search-parameter-domains}
\end{table}

We choose values by their position in the cost distribution. A probability $q$ gives the value $F^{-1}(q)$, so $q=1/2$ selects the median. The grid combines 193 equally spaced probabilities from $0.1$ to $0.9$ with 96 geometrically spaced probabilities from $10^{-12}$ to $0.1$ and their upper-tail reflections $1-q$. Geometric spacing places points across successive orders of magnitude in the tails. After duplicate endpoints are removed, there are 383 probabilities. For distributions with support $\mathbb R$, we check all $383(384)/2=73{,}536$ ordered pairs, including equal values.

For nonnegative costs, the argument in \eqref{eq:trivial_case} gives $D(x,y)\ge0$ whenever $x\le\psi(y)$. The remaining interval is parameterized by
\[
x=\psi(y)+t[y-\psi(y)],\qquad 0\le t\le1.
\]
For each of the 383 choices of $y$, we use 181 equally spaced $t$-values on $[0,0.99]$ and 28 additional values whose distances $1-t$ are geometrically spaced from $10^{-6}$ to $10^{-2}$. Removing one duplicate leaves 208 values; we also check $t=1$. This gives $80{,}047$ checks per shape, or $2{,}481{,}457$ over 31 shapes. The probability window bounds $y$; the $t$ parameterization covers the entire remaining interval for $x$.

\subsection{Searching Between Grid Points}

We supplement the grid with differential evolution \citep{storn1997differential}. This method maintains a collection of candidate points, proposes new points by combining and perturbing existing candidates, and keeps changes that lower the objective. Multiple runs reduce dependence on any one starting collection. We use seeds $0,\ldots,5$ when no shape parameter remains and $0,\ldots,9$ when shape is searched jointly with worker values. Each run uses population multiplier 18, at most 280 generations, relative tolerance $10^{-9}$, and a final local polishing step. Positive shape parameters are searched on a logarithmic scale. All probability coordinates lie in $[10^{-12},1-10^{-12}]$.

The deterministic grid and the optimizer cover the same finite domains. Neither guarantees that every adverse region is found. In particular, the search leaves unexamined more extreme tails, shapes outside the displayed ranges, and points missed within the continuous domains.

\subsection{Results and Accuracy Checks}

\Cref{tab:quadrature-screening} reports initial screening flags and how many survive more accurate quadrature. For Gumbel maxima, all 7,887 initial flags disappear at both higher orders; selected candidates are also positive at 100-digit precision. For Weibull and Gumbel minima, many flags persist and selected distinct pairs have negative $D$ at high precision.

\begin{table}[H]
\centering
\small
\setlength{\tabcolsep}{4pt}
\begin{tabular}{@{}lrrrrr@{}}
\hline
\textbf{Family} & \textbf{Grid checks} & \textbf{32 points} & \textbf{64 points} & \textbf{96 points} & \textbf{Equal values} \\
\hline
Normal & 73,536 & 0 & 0 & 0 & 0 \\
Log-normal & 2,481,457 & 0 & 0 & 0 & 0 \\
Log-logistic & 2,481,457 & 0 & 0 & 0 & 0 \\
Gumbel maxima & 73,536 & 7,887 & 0 & 0 & 0 \\
Weibull & 2,481,457 & 258,168 & 253,577 & 253,577 & 2,407 \\
Gumbel minima & 73,536 & 58,840 & 56,994 & 56,994 & 306 \\
\hline
\end{tabular}
\caption{Numbers of screening flags below $-10^{-10}$. The 64- and 96-point columns recheck the points flagged initially. The last column uses the closed-form limit at $x=y$, separately from quadrature. Grid counts include that diagonal check.}
\label{tab:quadrature-screening}
\end{table}

For example, the following pairs give negative normalized differences under the original four-term expression at 100-digit precision:
\[
\begin{aligned}
\text{Weibull }(k=19.999999999757964):\quad
&(x,y)=(1.1790017888391056,1.1790027937567658),\\
&G=-0.1148099866340864867\ldots;\\
\text{Gumbel minima}:\quad
&(x,y)=(2.432508223496572,2.440584312244446),\\
&G=-0.1407433343056030667\ldots.
\end{aligned}
\]
These detections show that the same procedure used for the other families can find counterexamples. They establish suboptimality numerically at the selected points; the analytical argument in \Cref{sec:failure_proofs} establishes it independently.


\section{Performance When Greedy Is Suboptimal}
\label{sec:performance_protocol}

This appendix supplies the proofs and experiments used in the second part of \Cref{sec:simulations}. We first prove the mean--median condition and its implications for the three cost families. We then explain how to evaluate welfare for a given set of worker values, reduce the Beta search to a bounded region, and describe the limiting experiments for concentrated costs. The final subsection collects the search settings and numerical results.

\subsection{Why Counterexamples Exist}
\label{sec:failure_proofs}

\begin{lemma}[Mean--median condition]
\label{lem:mean-median-obstruction}
Let $F$ be continuous with finite first moment and mean $\mu$. If $F(\mu)<1/2$, then some two-worker instance makes greedy ranking strictly suboptimal under welfare-optimal wages.
\end{lemma}

\proof{Proof.}
Fix $d>0$ and let the worker values be $x=y-d$ and $y$. As $y\to\infty$, integrability gives $\psi(y)=y-\mu+o(1)$ and $\psi(y-d)=y-d-\mu+o(1)$. Substituting in \eqref{eq:swap_difference} and using the continuity and $1$-Lipschitz property of $\psi$ gives
\begin{equation}
\lim_{y\to\infty}D(y-d,y)
=\psi(\mu+d)-\psi(\mu-d)-d
\defeq\kappa(d).
\label{eq:mean-median-kappa}
\end{equation}
Because $\psi'=F$, we have $\kappa(0)=0$ and $\kappa'(0)=2F(\mu)-1<0$. Thus $\kappa(d)<0$ for all sufficiently small positive $d$, and $D(y-d,y)<0$ for all sufficiently large $y$. The reverse order then yields higher welfare.
\hfill\Halmos
\endproof

The proof explains the construction in the main text: raise both worker values while keeping their difference small. Under the greedy order, the first wage $y-\psi(y-d)$ approaches $\mu+d$. If $F(\mu)<1/2$, that early offer is accepted less often than the offer to the fallback worker, conditional on reaching the pair, when $d$ is small. Reversing the ranking and reoptimizing wages can improve welfare.

\begin{proposition}[Beta$(\alpha,1)$]
\label{prop:beta-alpha-one-phase}
For $Z\sim\operatorname{Beta}(\alpha,1)$, greedy ranking is globally welfare-optimal for every number of workers and every value vector if and only if $0<\alpha\le1$. Every $\alpha>1$ admits a strict two-worker counterexample.
\end{proposition}
\proof{Proof.}
For $0<\alpha\le1$, the density $f(z)=\alpha z^{\alpha-1}$ is non-increasing and convex on $(0,1)$, so \Cref{thm:global_optimality} applies, including its permitted endpoint behavior. For $\alpha>1$, the mean is $\mu_\alpha=\alpha/(\alpha+1)$ and
\begin{equation}
F_\alpha(\mu_\alpha)=\left(\frac{\alpha}{\alpha+1}\right)^\alpha<\frac12.
\label{eq:beta-alpha-one-mean-cdf}
\end{equation}
Indeed, this expression equals $1/2$ at $\alpha=1$ and its logarithmic derivative is $-\log(1+1/\alpha)+1/(\alpha+1)<0$. Apply \Cref{lem:mean-median-obstruction}.
\hfill\Halmos
\endproof

For unit-scale Weibull costs, the mean and median are $\Gamma(1+1/k)$ and $(\log2)^{1/k}$. Hence the condition is $[\Gamma(1+1/k)]^k<\log2$. The function $r\mapsto\log\Gamma(1+r)/r$ is strictly increasing because $\log\Gamma(1+r)$ is strictly convex and vanishes at zero. Therefore $[\Gamma(1+1/k)]^k$ decreases strictly in $k$, from one at $k=1$ to $e^{-\gamma_{\rm E}}<\log2$ as $k\to\infty$. The mean and median cross once, at $k_0=3.439540619\ldots$. The lemma proves suboptimality for every $k>k_0$.

For standard Gumbel minima, $F(z)=1-e^{-e^z}$ and $\E[Z]=-\gamma_{\rm E}<\log(\log2)=\operatorname{median}(Z)$, where $\gamma_{\rm E}$ is Euler's constant. The lemma again applies. Location and positive scale do not affect either conclusion.

\subsection{Evaluating Greedy and the Best Ranking}

For a fixed order, begin with continuation value zero after the last worker. Working backward, a worker with net value $v_i$ changes the continuation value from $q$ to $q+\psi(v_i-q)$. Taking the offering order to be decreasing in worker values and evaluating it backward gives $\gre(\vv)$.

To compute the online optimum, let $Q(S)$ be the best welfare obtainable from the subset $S$ of available workers. Trying each worker as the first offer gives
\begin{equation}
Q(S)=\max_{i\in S}\left\{
Q(S\setminus\{i\})+\psi\bigl(v_i-Q(S\setminus\{i\})\bigr)
\right\},\qquad Q(\varnothing)=0.
\label{eq:performance-subset-dp}
\end{equation}
The update is nondecreasing in the downstream value because $\psi$ is $1$-Lipschitz. It is therefore optimal to use the best order for the remaining subset. Computing $Q$ from smaller to larger subsets accounts for every possible ranking, without separately listing all $n!$ orders. The result for the full worker pool is $\opt(\vv)$.

For a fixed cost distribution, the true worst-case ratio is
\begin{equation}
R_n(F)=\inf_{\vv:\,\opt_F(\vv)>0}\frac{\gre_F(\vv)}{\opt_F(\vv)}.
\label{eq:family-worst-ratio}
\end{equation}
The experiments search for vectors with low ratios. The ranking optimization in \eqref{eq:performance-subset-dp} is exhaustive, while the search over worker values may miss better candidates. We recompute every selected minimum with the same recursion at 80-digit precision. Thus each reported ratio is an upper bound on the unknown $R_n(F)$, subject to numerical evaluation; \Cref{thm:worst_case} supplies the rigorous lower bound $n/(2n-1)$.
\subsection{Beta Costs: A Compact Search and a Concentrated-Cost Limit}
\label{sec:beta-alpha-one}

We use Beta$(\alpha,1)$ on $[0,1]$. Multiplying costs and values by the same positive constant rescales welfare and leaves every ratio unchanged. The CDF and incremental value function are
\begin{equation}
F_\alpha(z)=z^\alpha\quad(0\le z\le1),\qquad
\label{eq:beta-alpha-one-cdf}
\end{equation}
\begin{equation}
\psi_\alpha(t)=
\begin{cases}
0, & t\le0,\\[2pt]
t^{\alpha+1}/(\alpha+1), & 0<t<1,\\[2pt]
t-\alpha/(\alpha+1), & t\ge1.
\end{cases}
\label{eq:beta-alpha-one-psi}
\end{equation}
For two workers with $x\le y$, denote the welfare of the greedy order and the reverse order by
\begin{align}
G_\alpha(x,y)&=\psi_\alpha(x)+\psi_\alpha(y-\psi_\alpha(x)),
\label{eq:beta-alpha-one-greedy-two}\\
A_\alpha(x,y)&=\psi_\alpha(y)+\psi_\alpha(x-\psi_\alpha(y)).
\label{eq:beta-alpha-one-alt-two}
\end{align}
The online optimum is $O_\alpha=\max\{G_\alpha,A_\alpha\}$, and we write
\begin{equation}
R_2(\alpha)=\inf_{x\le y:\,O_\alpha(x,y)>0}G_\alpha(x,y)/O_\alpha(x,y).
\label{eq:beta-alpha-one-ratio}
\end{equation}
The following reduction makes it sufficient to search a bounded region, even though worker values themselves are unrestricted.

\begin{lemma}[Compact two-worker search]
\label{lem:beta-alpha-one-compact}
Let $\mu_\alpha=\alpha/(\alpha+1)$. For $\alpha>1$, the infimum defining $R_2(\alpha)$ is unchanged if the search is restricted to
\begin{equation}
0\le y\le1+\mu_\alpha,\qquad \psi_\alpha(y)<x\le y.
\label{eq:beta-alpha-one-compact-domain}
\end{equation}
Every counterexample outside this region has a counterpart inside it with a weakly lower welfare ratio.
\end{lemma}

\proof{Proof.}
If $x\le\psi_\alpha(y)$, then the final term in
\eqref{eq:beta-alpha-one-alt-two} is zero.  Since $\psi_\alpha$ is
$1$-Lipschitz,
\[
\psi_\alpha(y)-\psi_\alpha\bigl(y-\psi_\alpha(x)\bigr)
\le\psi_\alpha(x),
\]
so $G_\alpha(x,y)\ge A_\alpha(x,y)$.  Hence strict failure requires
$x>\psi_\alpha(y)$, and therefore $x,y\ge0$.

Suppose $y>1+\mu_\alpha$.  Then
$x>\psi_\alpha(y)=y-\mu_\alpha>1$.  Write $d=y-x$, so
$0\le d<\mu_\alpha$.  Using the upper branch of
\eqref{eq:beta-alpha-one-psi},
\begin{align*}
G_\alpha(x,y)
&=x-\mu_\alpha+\psi_\alpha(\mu_\alpha+d),\\
A_\alpha(x,y)
&=x+d-\mu_\alpha+\psi_\alpha(\mu_\alpha-d).
\end{align*}
Thus $A_\alpha-G_\alpha$ depends only on $d$.  When this advantage is
positive, $G_\alpha/A_\alpha$ is increasing in $x$.  Translating the
pair down to $(1,1+d)$ preserves the advantage and weakly lowers the
ratio, and $1+d\le1+\mu_\alpha$.
\hfill\Halmos
\endproof

\paragraph{An explicit loss as costs concentrate.}
The universal guarantee bounds the ratio below. The following construction shows a two-worker loss of about $8.4\%$ in the limit as $\alpha\to\infty$.
\begin{proposition}[An explicit limiting counterexample]
\label{prop:beta-alpha-one-gap}
For every $\alpha>1$,
\begin{equation}
\frac23\le R_2(\alpha)<1.
\label{eq:beta-alpha-one-two-thirds}
\end{equation}
Furthermore, set
\begin{equation}
x_\alpha=1-\frac{9}{20\alpha},
\qquad
y_\alpha=1+\frac{37}{100\alpha}.
\label{eq:beta-alpha-one-explicit-pair}
\end{equation}
For this explicit two-worker instance,
\begin{equation}
\lim_{\alpha\to\infty}
\frac{G_\alpha(x_\alpha,y_\alpha)}
{O_\alpha(x_\alpha,y_\alpha)}
=\rho
\defeq
\frac{e^{-9/20}+\exp\!\left(37/100-e^{-9/20}\right)}
{137/100+e^{-91/50}}
\approx0.915663731.
\label{eq:beta-alpha-one-explicit-limit}
\end{equation}
Consequently, $\limsup_{\alpha\to\infty}R_2(\alpha)\le\rho$:
greedy can lose at least $1-\rho\approx8.4336\%$ with only two
workers.
\end{proposition}

\proof{Proof.}
For $n=2$, \Cref{sec:worst_case_guarantees} gives
$G_\alpha/\off\ge2/3$.  Since $O_\alpha\le\off$,
$G_\alpha/O_\alpha\ge2/3$.  The strict upper inequality follows from
\Cref{prop:beta-alpha-one-phase}.

For the explicit construction, define
\begin{equation}
h(s)=
\begin{cases}
e^s, & s<0,\\
1+s, & s\ge0.
\end{cases}
\label{eq:beta-alpha-one-boundary-h}
\end{equation}
Direct substitution in \eqref{eq:beta-alpha-one-psi} gives, for every
fixed $s$,
\begin{equation}
\alpha\psi_\alpha(1+s/\alpha)\longrightarrow h(s).
\label{eq:beta-alpha-one-boundary-limit}
\end{equation}
Let $u=-9/20$ and $v=37/100$.  Applying the limit first to the
terminal worker and then to the upstream worker gives
\begin{align*}
\alpha G_\alpha(x_\alpha,y_\alpha)
&\longrightarrow h(u)+h\bigl(v-h(u)\bigr),\\
\alpha A_\alpha(x_\alpha,y_\alpha)
&\longrightarrow h(v)+h\bigl(u-h(v)\bigr).
\end{align*}
The two limits are
\[
e^u+e^{v-e^u}\approx1.402820415,
\qquad
1+v+e^{u-1-v}\approx1.532025751,
\]
whose ratio is \eqref{eq:beta-alpha-one-explicit-limit}.
\hfill\Halmos
\endproof

\paragraph{Why zoom in near the upper endpoint?}
As $\alpha$ grows, costs concentrate near one. To keep a meaningful comparison, set worker values to $v_i^{(\alpha)}=1+s_i/\alpha$ and measure welfare in units of $1/\alpha$. This examines a shrinking neighborhood of the upper endpoint, often called a boundary layer. The limit $h$ in \eqref{eq:beta-alpha-one-boundary-h} retains the relevant uncertainty: $\alpha(1-Z)$ converges to an Exponential variable with mean one, since for fixed $t\ge0$,
\[
\Pr\{\alpha(1-Z)>t\}=(1-t/\alpha)^\alpha\longrightarrow e^{-t}.
\]
Thus the rescaled cost is the negative of an Exponential variable, and $h(s)$ is its incremental welfare function.

For a fixed order $\sigma$, define $q_{n+1}=0$ and
\begin{equation}
q_j=q_{j+1}+h(s_{\sigma(j)}-q_{j+1}),\qquad j=n,\ldots,1.
\label{eq:beta-alpha-one-limit-recursion}
\end{equation}
The functions $s\mapsto\alpha\psi_\alpha(1+s/\alpha)$ are $1$-Lipschitz. Their pointwise convergence in \eqref{eq:beta-alpha-one-boundary-limit} therefore permits backward induction with converging continuation values, giving
\begin{equation}
\alpha V_\alpha(v_{\sigma(1)}^{(\alpha)},\ldots,v_{\sigma(n)}^{(\alpha)})\longrightarrow q_1.
\label{eq:beta-alpha-one-limit-value}
\end{equation}
Because there are finitely many rankings, their maximum converges as well. The limiting online optimum is computed by
\begin{equation}
Q(S)=\max_{i\in S}\{Q(S\setminus\{i\})+h(s_i-Q(S\setminus\{i\}))\},\quad Q(\varnothing)=0.
\label{eq:beta-alpha-one-subset-dp}
\end{equation}
Every fixed limiting vector with positive optimal welfare can therefore be approximated by finite-$\alpha$ instances with converging ratios. This justifies the limiting Beta curve in \Cref{fig:failure_by_horizon}; it does not identify the worst ratio over all values or justify interchanging the limit with an unrestricted infimum.

The explicit pair in \Cref{prop:beta-alpha-one-gap} has limiting ratio $91.5663731\ldots\%$. Optimizing the limiting two-worker problem numerically gives the slightly lower $91.5661\ldots\%$. They agree to the three decimal places used in the tables, but only the former is the ratio of the displayed analytical construction.
\subsection{Weibull Costs: The Gumbel Limit}
Large-shape Weibull costs also concentrate near one, but have no finite upper endpoint. Centering costs at one and multiplying by $k$ gives a Gumbel-for-minima limit with support on both sides of zero. The result below shows that the welfare comparison converges along with the costs.

\begin{proposition}[Weibull--Gumbel boundary layer]
\label{prop:weibull-gumbel-limit}
Let $Z_k$ be unit-scale Weibull with shape $k$, and set $T_k=k(Z_k-1)$. Then $T_k$ converges in distribution to a standard Gumbel-for-minima random variable $T$. Moreover, for every fixed value vector $\bm{s}=(s_1,\ldots,s_n)$ and every ranking $\sigma$,
\begin{equation}
kV_{F_k}\!\left(1+\frac{s_{\sigma(1)}}k,\ldots,
1+\frac{s_{\sigma(n)}}k\right)
\longrightarrow V_H(s_{\sigma(1)},\ldots,s_{\sigma(n)}).
\label{eq:weibull-gumbel-recursion-limit}
\end{equation}
The convergence also holds after maximizing over the finitely many rankings. Consequently, every strict Gumbel-minimum counterexample generates Weibull counterexamples for all sufficiently large $k$.
\end{proposition}

\proof{Proof.}
For fixed $t$ and all sufficiently large $k$,
\[
\Pr(T_k\le t)=1-\exp\!\left[-(1+t/k)^k\right]
\longrightarrow 1-e^{-e^t}=H(t).
\]
Define $h_k(s)=k\psi_k(1+s/k)=\E[(s-T_k)^+]$. For $-k<t<0$,
$\Pr(T_k\le t)\le(1+t/k)^k\le e^t$, while the probability is zero for $t\le-k$; this integrable lower-tail bound and pointwise convergence give
$h_k(s)\to h(s)\defeq\E[(s-T)^+]$.

For a fixed order, multiply every step of the continuation recursion by $k$. If $q_{j+1}^{(k)}$ is the scaled downstream value, then
\[
q_j^{(k)}=q_{j+1}^{(k)}
+h_k\!\left(s_{\sigma(j)}-q_{j+1}^{(k)}\right),
\qquad q_{n+1}^{(k)}=0.
\]
The functions $h_k$ are $1$-Lipschitz, so backward induction gives \eqref{eq:weibull-gumbel-recursion-limit}. Taking the maximum preserves convergence because there are only $n!$ rankings. A strict limiting inequality persists for all sufficiently large $k$.
\hfill\Halmos
\endproof

\subsection{Search Settings and Numerical Results}
\label{sec:performance_design}

\paragraph{Two-worker Beta experiments.}
For each shape, we search the compact region in \Cref{lem:beta-alpha-one-compact} using
\[
0\le y\le1+\mu_\alpha,\qquad
x=\psi_\alpha(y)+t[y-\psi_\alpha(y)],\qquad 0\le t\le1.
\]
We evaluate an $801\times801$ equally spaced grid in $(y,t)$, then improve the best point by taking nearby steps in either coordinate and halving the step sizes when no improvement is found. Refinement stops when both relative step sizes are below $10^{-12}$. The 19 shapes above one are
\[
\begin{gathered}
1.0001,\ 1.001,\ 1.01,\ 1.02,\ 1.05,\ 1.1,\ 1.2,\ 1.3,\ 1.5,\ 2,\\
3,\ 5,\ 10,\ 20,\ 50,\ 100,\ 200,\ 500,\ 1000.
\end{gathered}
\]
Shapes at or below one use the analytical optimality result. The numerical search uses 12,190,419 grid evaluations and 24,571 refinement evaluations.

\paragraph{Two-worker Weibull experiments.}
We search $0\le y\le F_k^{-1}(1-10^{-12})$ and the same remaining band $x=\psi_k(y)+t[y-\psi_k(y)]$. The grid uses 241 equally spaced probabilities on $[10^{-4},1-10^{-4}]$, 41 geometric probabilities from $10^{-12}$ to $10^{-4}$, and their upper-tail reflections. This gives 321 distinct $y$-values. We combine 321 equally spaced $t$-values on $[0,1]$ with 41 additional values having geometric distances $1-t$ from $10^{-12}$ to $10^{-4}$, giving 362 values and 116,202 pairs per shape. The numerical shapes are
\[
1.5,\ 2,\ 3,\ k_0,\ 3.5,\ 4,\ 5,\ 7.5,\ 10,\ 15,\ 20,\ 30,\ 50,\ 100.
\]
The plotted shapes $0.5$ and $1$ use \Cref{thm:global_optimality}. Each numerical shape also uses six differential-evolution runs, with seeds $0,\ldots,5$, population multiplier 15, at most 400 generations, and bounded Powell refinement. Powell refinement improves a candidate through successive searches along coordinate directions and combinations of directions. Across shapes, there are 1,626,828 grid evaluations and 74,953 optimization or refinement evaluations.

\paragraph{Experiments with two through eight workers.}
We search Weibull with $k=10$, standard Gumbel minima, and the limiting Beta model in \eqref{eq:beta-alpha-one-limit-recursion}. Each coordinate lies in
\[
\begin{array}{ll}
\text{Weibull:} & [0,F_{10}^{-1}(1-10^{-12})],\\
\text{Gumbel minima:} & [F^{-1}(10^{-10}),F^{-1}(1-10^{-12})],\\
\text{Limiting Beta:} & [-4,2]\quad\text{for the scaled value }s_i.
\end{array}
\]
Each family and worker count uses seeds $0,\ldots,5$, population multiplier ten, at most 250 differential-evolution generations, and bounded Powell refinement. Both the two-worker Weibull optimization and these experiments use relative tolerance $10^{-10}$ and absolute tolerance $10^{-12}$ for differential evolution. Powell refinement uses coordinate tolerance $10^{-11}$, objective tolerance $10^{-13}$, and at most 250 iterations.

We also reuse good candidates across experiments. Every subset of each selected longer vector is evaluated at the corresponding smaller worker count, and the six best distinct projections are refined. The best shorter vectors are extended to initialize later searches; for Weibull, appending zero-value workers preserves their welfare ratios. Gumbel candidates also initialize the Weibull search through the scaling in \Cref{prop:weibull-gumbel-limit}. These checks help avoid missing a pattern already found elsewhere, but do not establish global optimality of the search.

The Weibull and Gumbel experiments use 542,659 initial optimization or refinement evaluations and 23,441 further refinement evaluations across worker counts. The limiting Beta experiment uses 281,066 and 72,699, respectively. The saved performance results contain 28 selected 80-digit checks for Weibull/Gumbel and 26 for finite and limiting Beta.

\begin{table}[H]

\centering
\small
\setlength{\tabcolsep}{4pt}
\begin{tabular}{@{}lrrrrrrrr@{}}
\hline
$\alpha$ & $1$ & $1.1$ & $2$ & $5$ & $10$ & $100$ & $1000$ \\
\hline
Beta ratio (\%) & $100.000$ & $99.537$ & $96.821$ & $93.982$ & $92.833$ & $91.698$ & $91.579$ \\
\hline
\end{tabular}
\par\medskip
\begin{tabular}{@{}lrrrrrrrr@{}}
\hline
$k$ & $1$ & $k_0$ & $3.5$ & $5$ & $10$ & $20$ & $50$ & $100$ \\
\hline
Weibull ratio (\%) & $100.000$ & $100.000$ & $99.981$ & $99.167$ & $97.859$ & $97.100$ & $96.615$ & $96.448$ \\
\hline
\end{tabular}
\caption{Selected lowest two-worker ratios found. The entries at $\alpha\le1$ and $k\le1$ are analytical. The entry at $k_0$ reports a finite search with no counterexample found.}
\label{tab:beta-alpha-one-performance}
\label{tab:weibull-shape-performance}
\end{table}

\begin{figure}[H]

\centering
\includegraphics[width=0.55\linewidth]{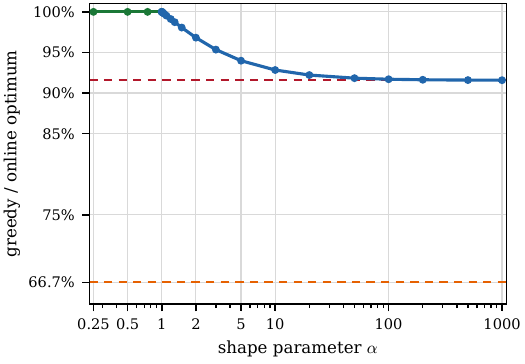}
\caption{Two-worker Beta performance as costs concentrate near their upper endpoint. Green denotes proved optimality and blue shows the lowest ratios found. The dashed red line is the ratio of the explicit limiting construction in \Cref{prop:beta-alpha-one-gap}, not a proof of the optimal limiting ratio. The orange line is the universal two-worker guarantee.}
\label{fig:beta-alpha-one-performance}
\label{fig:beta-alpha-one-by-alpha}
\end{figure}

\begin{table}[H]

\centering
\small
\setlength{\tabcolsep}{4pt}
\begin{tabular}{@{}lrrrrrrr@{}}
\hline
Number of workers $n$ & $2$ & $3$ & $4$ & $5$ & $6$ & $7$ & $8$ \\
\hline
Weibull $k=10$ & $97.859$ & $97.625$ & $97.625$ & $97.625$ & $97.625$ & $97.625$ & $97.625$ \\
Gumbel minima & $96.279$ & $95.363$ & $95.081$ & $95.081$ & $95.081$ & $95.069$ & $95.000$ \\
Limiting Beta & $91.566$ & $90.004$ & $89.377$ & $89.041$ & $88.833$ & $88.692$ & $88.589$ \\
Universal guarantee & $66.667$ & $60.000$ & $57.143$ & $55.556$ & $54.545$ & $53.846$ & $53.333$ \\
\hline
\end{tabular}
\caption{Lowest ratios found as the number of workers increases, in percent. The first three rows are numerical results. The last row is the rigorous bound $100n/(2n-1)$. The Beta row uses the limiting model, rather than a fixed finite shape.}
\label{tab:suboptimal-horizon-performance}
\end{table}


\section{A Unified Quantile Framework for Sequential Offering}
\label{sec:quantile_framework}

This section places welfare, profit, and weighted objectives in a common
acceptance-probability formulation. If the platform targets acceptance
probability $q\in[0,1]$ for worker $i$, write its expected one-step
reward as
\begin{equation}
R_i(q)=v_iq-K(q).
\end{equation}
For the physical cost distribution $F$, the three cases of interest are
\begin{align}
K_W(q)&=\int_0^q F^{-1}(s)\,ds, &
K_\Pi(q)&=qF^{-1}(q), \notag\\
K_\alpha(q)&=(1-\alpha)K_W(q)+\alpha K_\Pi(q),
& &\alpha\in[0,1].                                      \label{eq:K_objectives}
\end{align}

\begin{assumption} \label{ass:K_convex}
$K$ is finite and four times continuously differentiable on $(0,1)$, with $K''(q)>0$. Its endpoint values are the one-sided limits: $K(0)=\lim_{q\downarrow0}K(q)=0$ and $K(1)=\lim_{q\uparrow1}K(q)\in\R\cup\{+\infty\}$. The one-sided limits of $K'$ may be infinite.
\end{assumption}

Fix a sequence and let $x$ be the current value less the downstream
continuation value. The incremental value is the convex conjugate
\begin{equation}
\psi_K(x)=\max_{q\in[0,1]}\{qx-K(q)\}.
\end{equation}
At an interior optimum, $x=K'(q)$ and
$\psi_K'(x)=q=(K')^{-1}(x)$. At the endpoints, the maximizer is
$q=0$ or $q=1$, which represents skipping an offer or inducing certain
acceptance. The dynamic program therefore has exactly the recursion in
\Cref{lem:optimal_wages}, with $\psi$ replaced by $\psi_K$.

This equivalence can be stated probabilistically. Define an effective
cost $Y_K$ by its quantile function
\begin{equation}
G_K^{-1}(q)=K'(q),\qquad q\in(0,1).
\end{equation}
Up to an additive constant fixed by $K(0)=0$,
\begin{equation}
\psi_K(x)=\E[(x-Y_K)^+].                                  \label{eq:effective_cost}
\end{equation}
Thus the generalized objective is isomorphic to the welfare problem for
the effective distribution $G_K$; it is not merely analogous to it.

\begin{theorem} \label{thm:generalized_optimality}
Under \Cref{ass:K_convex}, suppose the endpoint behavior of $G_K$
satisfies the support conditions in \Cref{thm:global_optimality}. The
greedy ranking is globally optimal if, for every $q\in(0,1)$,
\begin{equation}
K'''(q)\ge0,
\qquad
3\big(K'''(q)\big)^2-K^{(4)}(q)K''(q)\ge0.               \label{eq:K_conditions}
\end{equation}
\end{theorem}
\proof{Proof.}
The effective density satisfies
\[
g_K(K'(q))=\frac{1}{K''(q)}.
\]
Differentiating with respect to effective cost gives
\begin{align*}
g_K'(K'(q))
&=-\frac{K'''(q)}{(K''(q))^3},\\
g_K''(K'(q))
&=\frac{3(K'''(q))^2-K^{(4)}(q)K''(q)}{(K''(q))^5}.
\end{align*}
The conditions in \eqref{eq:K_conditions} therefore say exactly that
$g_K$ is non-increasing and convex. Apply
\Cref{thm:global_optimality} to the effective-cost representation
\eqref{eq:effective_cost}.
\hfill\Halmos
\endproof

For welfare, $K_W'(q)=F^{-1}(q)$, so the effective distribution is the
physical distribution itself. For profit, let $w=F^{-1}(q)$ and define
the virtual cost
\begin{equation}
\phi(w)=w+\frac{F(w)}{f(w)}.
\end{equation}
Then $K_\Pi'(q)=\phi(w)$. Under the regularity condition $\phi'(w)>0$,
the effective profit cost is $Y_\Pi=\phi(Z)$ and has density
\begin{equation}
g_\Pi(\phi(w))=\frac{f(w)}{\phi'(w)}.                    \label{eq:profit_density}
\end{equation}
Consequently, greedy ranking is profit-optimal whenever this
virtual-cost density is non-increasing and convex and has the endpoint
behavior required by \Cref{thm:global_optimality}. The interior pricing
rule is
\[
\phi(w_i^*)=v_i-\text{(downstream continuation value)}.
\]
Margins below the lowest virtual cost are skipped, while margins above a finite highest virtual cost induce acceptance probability one.

One useful, but limited, physical-density implication is the following.
If $f$ is non-increasing and log-concave, then $K_\Pi''>0$ and
$K_\Pi'''\ge0$, so the effective profit density is non-increasing.
If its lower support endpoint is finite, this supports local stability after translating that endpoint to zero.
It does \emph{not} imply that the effective density is convex; global optimality still requires the second inequality in \eqref{eq:K_conditions} or a direct verification of \eqref{eq:profit_density}.
Indeed, $K_\Pi''=2/f-Ff'/f^3$ and $K_\Pi'''=-3f'/f^3+F(3(f')^2-ff'')/f^5$, evaluated at $w=F^{-1}(q)$. The inequalities $f'\le0$ and $ff''\le(f')^2$ give the stated signs.

The next family gives a closed-form global result simultaneously for
welfare, profit, and every convex combination between them.

\begin{theorem} \label{thm:profit_canonical}
Let the physical cost have a truncated-Exponential density
\[
f(w)=\frac{\beta e^{-\beta w}}{1-e^{-\beta B}},
\qquad 0\le w\le B,
\]
where $\beta>0$ and $B\in(0,\infty]$. This includes the Exponential
law when $B=\infty$; the limit $\beta\downarrow0$ for finite $B$ is
Uniform on $[0,B]$. For every $\alpha\in[0,1]$, greedy ranking is
globally optimal for the objective $K_\alpha$ in
\eqref{eq:K_objectives}.
\end{theorem}
\proof{Proof.}
Let $c=1-e^{-\beta B}$, with $c=1$ when $B=\infty$, and set
$s=1-cq$. Since $F^{-1}(q)=-\beta^{-1}\log s$, direct differentiation
of \eqref{eq:K_objectives} gives
\begin{align*}
K_\alpha'(q)
&=\frac{1}{\beta}\left[-\log s+\alpha\left(\frac1s-1\right)\right],\\
K_\alpha''(q)
&=\frac{c}{\beta}\frac{s+\alpha}{s^2},\\
K_\alpha'''(q)
&=\frac{c^2}{\beta}\frac{s+2\alpha}{s^3},\\
K_\alpha^{(4)}(q)
&=\frac{2c^3}{\beta}\frac{s+3\alpha}{s^4}.
\end{align*}
Thus $K_\alpha''>0$, $K_\alpha'''\ge0$, and
\[
3(K_\alpha''')^2-K_\alpha^{(4)}K_\alpha''
=\frac{c^4}{\beta^2s^6}
  \big(s^2+4\alpha s+6\alpha^2\big)\ge0.
\]
The endpoint conditions follow from the displayed effective quantile.
When $\beta\downarrow0$ and $B<\infty$,
$K_\alpha(q)=B(1+\alpha)q^2/2$, so the same conclusion follows with
$K_\alpha'''=K_\alpha^{(4)}=0$.
\hfill\Halmos
\endproof

\section{Robustness to Deterministically Decaying Markets}
\label{sec:decaying_markets}

In the baseline model, private costs have a common stationary distribution. We now allow known deterministic deterioration across offer positions. At position $t$, whichever worker is approached has cost distribution $F_t$, independent of past rejections; observable net values and the available worker pool remain fixed. This can describe rising private costs as the delivery window shrinks. It does not model learning an unknown common cost distribution or changing the set of eligible workers.

\begin{definition}[Deterministically Decaying Market]
A market is deterministically decaying if $F_1(w)\ge F_2(w)\ge\cdots\ge F_n(w)$ for every real $w$. Thus, at a fixed wage, acceptance becomes weakly less likely at later positions.
\end{definition}

\begin{theorem}\label{thm:decaying_market}
Consider a deterministically decaying market and any real net values $v_1\ge\cdots\ge v_n$. Assume $F_1$ has a finite lower support endpoint. Suppose each downstream distribution $F_t$, $t\ge2$, has a non-negative stationary swap difference for every $x\le y$. This condition holds, in particular, under the hypotheses of \Cref{thm:global_optimality} after translating its lower endpoint, or under \Cref{thm:half_normal}. Then greedy ranking with optimal wages is globally optimal. No additional density-shape condition is needed for $F_1$.
\end{theorem}
\proof{Proof.}
If $L$ is a lower support bound for $F_1$, first-order stochastic dominance gives $F_t(w)=0$ for every $w<L$ and every $t$, so $\psi_t(x)=\int_{-\infty}^xF_t(s)\,ds$ is finite. The Bellman update is $z+\psi_t(u_t-z)$.

For an inverted adjacent pair, let $x\le y$ be its margins over the unchanged continuation value after the pair. Swapping to put the higher value first changes the pair's value by
\begin{align*}
D_t(x,y)={}&\psi_{t+1}(x)+\psi_t(y-\psi_{t+1}(x))\\
&-\psi_{t+1}(y)-\psi_t(x-\psi_{t+1}(y)).
\end{align*}
Let $D^*_{t+1}$ denote the stationary swap difference for $F_{t+1}$, and set $H_t=\psi_t-\psi_{t+1}$. The dominance assumption implies that $H_t$ is non-decreasing. Hence
\begin{align*}
D_t(x,y)-D^*_{t+1}(x,y)
={}&H_t(y-\psi_{t+1}(x))\\
&-H_t(x-\psi_{t+1}(y))\ge0,
\end{align*}
because the first argument is at least the second. By assumption $D^*_{t+1}\ge0$, so each inverted pair can be swapped without lowering its continuation value. Every preceding update $z\mapsto z+\psi_k(u_k-z)$ is non-decreasing, by the 1-Lipschitz property of $\psi_k$. Reoptimizing preceding wages therefore preserves the improvement. Repeated swaps give the greedy order.
\hfill\Halmos
\endproof

Earlier acceptance opportunities are better under first-order stochastic dominance. The displayed comparison shows that this advantage adds a non-negative term to the stationary swap gain. This is why deterioration preserves a downstream distribution's greedy-optimal ordering property and why the first distribution needs no separate shape restriction.

\end{APPENDICES}

\end{document}